\documentclass[letterpaper,11pt]{article}
\usepackage{graphicx}
\usepackage[margin=1in]{geometry}
\usepackage{setspace}
\usepackage{microtype} 
\usepackage{amssymb}
\usepackage{amsmath}
\usepackage{amsthm}
\usepackage{bm}
\usepackage{xspace}
\usepackage{epsfig}
\usepackage{ifpdf}
\usepackage{url,hyperref}
\usepackage{latexsym}
\usepackage{xspace}
\usepackage{color}
\usepackage{makeidx}
\usepackage{stackrel}
\usepackage{amscd}
\usepackage[all]{xy}
\usepackage{lineno}
\usepackage[section]{algorithm}
\usepackage[noend]{algpseudocode}
\algnewcommand{\LineComment}[1]{\Statex\hspace{\algorithmicindent}\(\triangleright\) #1}
\algnewcommand\algorithmicforeach{\textbf{for each}}
\algdef{S}[FOR]{ForEach}[1]{\algorithmicforeach\ #1\ \algorithmicdo}
\usepackage{varwidth}
\usepackage{framed}
\usepackage{pb-diagram}
\usepackage[dvipsnames]{xcolor}
\usepackage[labelformat=parens]{subfig}
\usepackage[symbol]{footmisc}
\usepackage{mathtools}
\usepackage{etoolbox} 
\usepackage[normalem]{ulem}
\usepackage{enumitem}

\makeatletter
\expandafter\patchcmd\csname\string\algorithmic\endcsname{\itemsep\z@}{\itemsep=0.25ex}{}{}
\makeatother

\long\def\remove#1{}

\theoremstyle{plain}
\newtheorem{theorem}{Theorem}[section]
\newtheorem{proposition}{Proposition}[section]

\theoremstyle{definition}
\newtheorem{Definition}{Definition}[section]
\newtheorem{problem}{Problem}[section]
\newtheorem{Remark}{Remark}[section]

\newcommand{\Hm}{\text{\sf H}}
\newcommand{\Chn}{\text{\sf C}}
\newcommand{\Zyc}{\text{\sf Z}}
\newcommand{\Bnd}{\text{\sf B}}

\newcommand{\by}{\times}

\renewcommand{\ker}{\mathrm{Ker}\,}
\newcommand{\img}{\mathrm{Img}\,}
\newcommand{\interior}{\mathrm{Int}}
\newcommand{\cof}{\mathrm{cof}\,}

\newcommand{\persdgm}{\mathsf{D}}
\newcommand{\ind}{\mathrm{ind}}
\newcommand{\bd}{\mathrm{bd}}

\newcommand{\Sspn}{\mathcal{S}}

\newcommand{\real}{\mathbb{R}}
\renewcommand{\bar}[1]{\overline{#1}}

\newcommand{\bll}{\mathbb{B}}

\newcommand{\inv}{^{-1}}

\newcommand\SetSymbol[1][]{\nonscript\:#1\vert\allowbreak\nonscript\:\mathopen{}}
\providecommand\given{} 
\DeclarePairedDelimiterX\Set[1]\{\}{\renewcommand\given{\SetSymbol[\delimsize]}#1}

\let\wdhat\widehat
\let\wdtild\widetilde

\let\emptyset\varnothing
\let\intersect\cap
\let\union\cup

\newcommand{\Acal}{\mathcal{A}}
\newcommand{\Bcal}{\mathcal{B}}

\newcommand{\Fcal}{\mathcal{F}}
\newcommand{\Ical}{\mathcal{I}}

\newcommand{\Mcal}{\mathcal{M}}
\newcommand{\Ncal}{\mathcal{N}}

\newcommand{\Rcal}{\mathcal{R}}

\newcommand{\Vcal}{\mathcal{V}}
\newcommand{\Zcal}{\mathcal{Z}}

\newcommand{\Gbb}{\mathbb{G}}

\newcommand{\Zbb}{\mathbb{Z}}

\newcommand{\aG}{\alpha}

\newcommand{\DG}{\Delta}

\newcommand{\kG}{\kappa}

\newcommand{\oG}{\omega}
\newcommand{\OG}{\Omega}

\newcommand{\sG}{\sigma}
\newcommand{\SG}{\Sigma}
\newcommand{\thG}{\theta}
\newcommand{\zG}{\zeta}

\newcommand{\Dim}{d}
\newcommand{\diml}{q}
\newcommand{\Dimtop}{{d+1}}
\newcommand{\Dimless}{d}
\newcommand{\Dimlsls}{{d-1}}
\newcommand{\birth}{{\beta}}
\newcommand{\death}{{\delta}}

\makeatletter
\newcommand{\algmargin}{\the\ALG@thistlm}
\makeatother
\newlength{\whilewidth}
\algdef{SE}[parWHILE]{parWhile}{EndparWhile}[1]
  {\parbox[t]{\dimexpr\linewidth-\algmargin}{%
     \hangindent\whilewidth\strut\algorithmicwhile\ #1\ \algorithmicdo\strut}}{\algorithmicend\ \algorithmicwhile}%
\algdef{SE}[parIF]{parIf}{EndparIf}[1]
  {\parbox[t]{\dimexpr\linewidth-\algmargin}{%
     \hangindent\whilewidth\strut\algorithmicif\ #1\ \algorithmicthen\strut}}{\algorithmicend\ \algorithmicif}%
\algnewcommand{\parState}[1]{\State%
  \parbox[t]{\dimexpr\linewidth-\algmargin}{\strut #1\strut}}

\begin{document}

\title{Computing Minimal Persistent Cycles:\\Polynomial and Hard Cases\thanks{Supported by NSF grants CCF-1740761 and CCF-1839252.}}

\author{Tamal K. Dey\thanks{Department of Computer Science and Engineering, The Ohio State University. \texttt{dey.8@osu.edu}}
\and Tao Hou\thanks{Department of Computer Science and Engineering, The Ohio State University. \texttt{hou.332@osu.edu}}
\and Sayan Mandal\thanks{Department of Computer Science and Engineering, The Ohio State University. \texttt{mandal.25@osu.edu}}
}

\date{}

\maketitle
\thispagestyle{empty}

\begin{abstract}
Persistent cycles, especially the minimal ones, are useful geometric features
functioning as augmentations for the intervals
in a purely topological persistence diagram (also termed as barcode).
In our earlier work,
we showed that 
computing minimal 1-dimensional persistent cycles (persistent 1-cycles) for finite intervals is NP-hard
while the same for infinite intervals is polynomially tractable.
In this paper, we address this problem for general dimensions with $\Zbb_2$ coefficients.
In addition to proving that it is NP-hard to compute minimal persistent $\Dim$-cycles ($\Dim>1$)
for both types of intervals given arbitrary simplicial complexes, 
we identify two interesting cases which are polynomially tractable.
These two cases assume the complex to
be a certain generalization of manifolds
which we term as {\it weak pseudomanifolds}.
For finite intervals from the $\Dim_\text{th}$ persistence diagram of a weak $(\Dim+1)$-pseudomanifold,
we utilize the fact that persistent cycles of such intervals are null-homologous
and reduce the problem to a minimal cut problem.
Since the same problem for infinite intervals is NP-hard,
we further assume the weak $(\Dim+1)$-pseudomanifold to be embedded in $\real^{\Dim+1}$
so that the complex has a natural dual graph structure
and the problem reduces to a minimal cut problem.
Experiments with both algorithms on scientific data 
indicate that the minimal persistent cycles capture various significant features of the data.
\end{abstract}

\newpage
\setcounter{page}{1}

\section{Introduction}
Persistent homology~\cite{edelsbrunner2010computational}, which captures essential topological features of data,
has proven to be a useful stable descriptor 
since Edelsbrunner~et~al.~\cite{edelsbrunner2000topological} first proposed the algorithm for its computation.
The understanding of topological persistence was later expanded by several works~\cite{chazal2009proximity,cohen2007stability,dey2014computing,zomorodian2005computing}
in terms of both theory and computation.
To make use of persistent homology, one typically computes a {\it persistence diagram} (also called {\it barcode})
which is a set of intervals with birth and death points.
Besides just utilizing the set of intervals, 
some applications~\cite{dey19pers1cyc,wu2017optimal} need persistence diagrams 
augmented with representative cycles for the intervals
for gaining more insight into the data.
These representative cycles, termed as {\it persistent cycles}~\cite{dey19pers1cyc},
have been studied by Wu~et~al.~\cite{wu2017optimal}, Obayashi~\cite{obayashi2018volume},
and Dey~et~al.~\cite{dey19pers1cyc} recently from the view-point of optimality.

Although the original persistence algorithm of Edelsbrunner et al.~\cite{edelsbrunner2000topological} implicitly computes persistent cycles, 
it does not necessarily provide minimal ones.
In an earlier work~\cite{dey19pers1cyc}, we showed that it is NP-hard to compute minimal persistent $1$-cycles
(cycles for 1-dimensional homology groups)
when the given interval is finite. Interestingly, the same 
for infinite intervals turned out to be computable in polynomial time~\cite{dey19pers1cyc}.
This naturally leads to the following questions:
Are there other interesting cases beyond $1$-dimension for which minimal persistent cycles can be computed in polynomial time? Also, what are the cases that are NP-hard?
In this paper, we settle the complexity question for computing 
minimal persistent cycles with $\Zbb_2$ coefficients 
in general dimensions.
We first show that when $\Dim\geq 2$,
computing minimal persistent $\Dim$-cycles for both finite and infinite intervals is NP-hard in general.
We then identify a special but important class of simplicial complexes, 
which we term as {\it weak $(\Dim+1)$-pseudomanifolds},
whose minimal persistent $\Dim$-cycles can be computed in polynomial time.
A~weak $(\Dim+1)$-pseudomanifold\footnote{
  The naming of weak pseudomanifold is adapted from the commonly accepted name {\it pseudomanifold} (see Definition~\ref{dfn:pseudomnfld}).} 
is a generalization of a $(\Dim+1)$-manifold and is defined as follows:

\begin{Definition} \label{dfn:weak-pseudo}
A simplicial complex $K$ is a {\it weak $(\Dim+1)$-pseudomanifold} if each $\Dim$-simplex is a face of no more than two $(\Dim+1)$-simplices in $K$.
\end{Definition}

Specifically, 
we find that if the given complex is a weak $(\Dim+1)$-pseudomanifold,
the problem of computing minimal persistent $\Dim$-cycles for finite intervals 
can be cast into a minimal cut problem (see Section~\ref{sec:fin-alg})
due to the fact that persistent cycles of such kind are null-homologous in the complex.
However, when $d\geq 2$ and intervals are infinite,
the computation of the same becomes NP-hard (see Section~\ref{sec:hardness}).
Nonetheless, for infinite intervals, if
we assume that the weak $(\Dim+1)$-pseudomanifold is embedded in $\real^{\Dim+1}$,
the minimal persistent cycle problem 
reduces to a minimal cut problem (see Section~\ref{sec:inf-alg}) and hence belongs to P. 
Note that a simplicial complex embedded in $\real^{\Dim+1}$ is automatically a weak $(\Dim+1)$-pseudomanifold.
Also note that while there is an algorithm~\cite{chen2011hardness} in the non-persistence setting
which computes minimal $\Dim$-cycles by minimal cuts,
the non-persistence algorithm assumes the $(\Dim+1)$-complex to be embedded in $\real^{\Dim+1}$.
Our algorithm for finite intervals, to the contrary, does not need the embedding assumption. 

In order to make our statements about the hardness results precise, 
we let PCYC-FIN$_\Dim$ denote 
the problem of computing minimal persistent $\Dim$-cycles for finite intervals 
when the given simplicial complex is arbitrary,
and let PCYC-INF$_\Dim$ denote the same problem for infinite intervals (see definitions of Problem~\ref{prob:pcyc-fin} and~\ref{prob:pcyc-inf}).
We also let WPCYC-FIN$_\Dim$ denote a subproblem\footnote{
  For two problems $P_1$ and $P_2$, $P_2$ is a {\it subproblem} of $P_1$ 
  if any instance of $P_2$ is an instance of $P_1$
  and $P_2$ asks for computing the same solutions as $P_1$.
} of PCYC-FIN$_\Dim$
and let WPCYC-INF$_\Dim$, WEPCYC-INF$_\Dim$ denote two subproblems of PCYC-INF$_\Dim$,
with the subproblems requiring additional constraints on the given simplicial complex.
Table~\ref{tbl:hardness} lists the hardness results for all problems of interest,
where the column ``Restriction on $K$'' specifies 
the additional constraints subproblems require on the given simplicial complex $K$.
Note that WPCYC-INF$_\Dim$ being NP-hard trivially implies that 
PCYC-INF$_\Dim$ is NP-hard.

\begin{table*}[htb!]
\caption{Hardness results for minimal persistent cycle problems with bold results denoting new findings.}
\label{tbl:hardness}
\small
\centering
\begin{tabular}{llll}
\hline
{\it Problem}     & {\it Restriction on $K$} & $d$                  & {\it Hardness} \\ \hline
PCYC-FIN$_\Dim$   & $-$                                             & $\geq 1$ & \textbf{NP-hard}    \\
WPCYC-FIN$_\Dim$  & $K$ a weak $(\Dim+1)$-pseudomanifold            & $\geq 1$ & \textbf{Polynomial} \\
PCYC-INF$_\Dim$   & $-$                                             & $=1$     & Polynomial     \\
WPCYC-INF$_\Dim$  & $K$ a weak $(\Dim+1)$-pseudomanifold            & $\geq 2$ & \textbf{NP-hard}    \\
WEPCYC-INF$_\Dim$ & $K$ a weak $(\Dim+1)$-pseudomanifold in $\real^{\Dim+1}$ & $\geq 2$ & \textbf{Polynomial} \\ \hline
\end{tabular}
\end{table*}

\paragraph{Main contributions.}
We summarize our contributions as follows:
\begin{itemize}
  \item We prove the NP-hardness of PCYC-FIN$_\Dim$ and WPCYC-INF$_\Dim$ for all $d\geq 2$.
  \item We present two polynomial time algorithms
  for WPCYC-FIN$_\Dim$ and WEPCYC-INF$_\Dim$ when $d\geq 1$,
  based on the duality of minimal persistent cycles and minimal cuts.
  Other than the minimal cut computation, 
  steps in both algorithms run in linear or almost linear time.
\end{itemize}


\subsection{Related works}\label{sec:review}
In the context of computing optimal cycles, most works have been done in the non-persistence setting. 
These works compute minimal cycles for homology groups of a given simplicial complex.
Only very few works address the problem while taking into account the persistence.
We review some of the relevant works below.

\paragraph{Minimal cycles for homology groups.}
In terms of computing minimal cycles for homology groups,
two problems are of most interest:
the localization problem and the minimal basis problem.
The localization problem asks for computing a minimal cycle in a homology class
and the minimal basis problem asks for computing
a set of generating cycles for a homology group
whose sum of weights is minimal.
With $\Zbb_2$ coefficients, these two problems are in general hard.
Specifically, Chambers~et~al.~\cite{chambers2009minimum}
proved that the localization problem over dimension one is NP-hard when the given simplicial complex is a 2-manifold.
Chen and Freedman~\cite{chen2011hardness} proved that
the localization problem is NP-hard to approximate with fixed ratio over arbitrary dimension. 
They also showed that the minimal basis problem is NP-hard to approximate with fixed ratio over dimension greater than one.
For one-dimensional homology,
Dey~et~al.~\cite{dey2010approximating} proposed a polynomial 
time algorithm for the minimal basis problem.
Several other works~\cite{borradaile2017minimum,chen2010measuring,dey2011optimal,erickson2005greedy}
address variants of the two problems while considering
special input classes, alternative cycle measures, or coefficients for homology other than $\Zbb_2$.




In this work, we use graph cuts and their duality extensively.
The duality of cuts on a planar graph and separating cycles on the dual graph 
has long been utilized to efficiently compute maximal flows and minimal cuts on planar graphs,
a topic for which Chambers~et~al.~\cite{chambers2009minimum} provide a comprehensive review.
In their paper~\cite{chambers2009minimum},
Chambers~et~al. discover the duality between minimal cuts of a surface-embedded graph
and minimal homologous cycles in a dual complex,
and then devise $O(n\log n)$ algorithms for both problems 
assuming the genus of the surface to be fixed.
Chen and Freedman~\cite{chen2011hardness} proposed an algorithm which computes a minimal non-bounding $\Dim$-cycle
given a $(\Dim+1)$-complex embedded in $\real^{\Dim+1}$,
utilizing a natural duality of $\Dim$-cycles in the complex and cuts in the dual graph.
The minimal non-bounding cycle algorithm can be further extended to 
solve the localization problem and the minimal basis problem
over dimension $\Dim$
given a $(\Dim+1)$-complex embedded in $\real^{\Dim+1}$.

\paragraph{Persistent cycle.}
As pointed out earlier, our main focus is the optimality of representative cycles in the
persistence framework.
Some early works~\cite{emmett2016multiscale,escolar2016optimal} 
address the representative cycle problem for persistence
by computing minimal cycles at the birth points of intervals 
without considering what actually die at the death points.
Wu~et~al.~\cite{wu2017optimal} proposed an algorithm computing minimal persistent 1-cycles
for finite intervals using an annotation technique and heuristic search. 
However, the time complexity of the algorithm is exponential in the worst-case.
Obayashi~\cite{obayashi2018volume} casts the minimal persistent cycle problem for finite intervals 
into an integer program,
but the rounded result of the relaxed linear program is not guaranteed to be optimal.
Dey~et~al.~\cite{dey19pers1cyc} formalizes the definition of persistent cycles for both finite and infinite intervals.
They also proved the NP-hardness of computing minimal persistent 1-cycles for finite intervals
and proposed a polynomial time algorithm for computing non-optimal ones which are still good in practice. 

\section{Preliminaries}\label{sec:pre}

In this section we present some concepts necessary for presenting the results in this paper.

\paragraph{Simplicial complex.}
A {\it simplicial complex} $K$ is a collection of {\it simplices}
which are abstractly defined as subsets of a ground set called the {\it vertex set} of $K$.
If a simplex $\sG$ is in $K$, then all its subsets called its {\it faces} 
are also in $K$. 
The simplex $\sG$ is also referred to as a $\diml$-simplex
if the cardinality of the vertex set of $\sG$ is $\diml+1$.
A {\it$\diml$-face} of $\sG$ is a $\diml$-simplex being a face of $\sG$
and a {\it$\diml$-coface} of $\sG$ is a $\diml$-simplex having $\sG$ as a face.
We call a $\diml$-simplex of $K$ a {\it boundary $\diml$-simplex} if it has less than two $(\diml+1)$-cofaces in $K$.
A {\it simplicial set} is a set of simplices 
and the {\it closure} of a simplicial set $\SG$ is the simplicial complex consisting of all the faces of the simplices in $\SG$.
A~simplicial complex is {\it finite} if it contains finitely many simplices.
In this paper, we only consider finite simplicial complexes.

If each vertex of a simplicial complex $K$ is a point in a Euclidean space,
then each simplex of $K$ can be interpreted as the convex hull of its vertices.
The simplicial complex $K$ is said to be {\it embedded in} the Euclidean space if 
the interiors of all its simplices are disjoint.
The {\it underlying space} of $K$, denoted by $|K|$, is the point-wise union of all the simplices of $K$.

\begin{Definition}[Oriented simplex~\cite{munkres2018elements}]
A $\diml$-simplex 
with an ordering of its vertices is an {\it oriented $\diml$-simplex}.
For each $\diml$-simplex $\sG$ ($\diml>0$),
there are exactly two equivalent classes of vertex orderings, 
resulting in two oriented $\diml$-simplices of $\sG$.
We refer to them as the {\it oppositely oriented $\diml$-simplices}.
\end{Definition}

\begin{Remark}
Any simplex by default is unoriented. 
We denote an unoriented $\diml$-simplex $\sG$ spanned by vertices 
$v_0,\ldots,v_\diml$ as $\sG=\Set{v_0,\ldots,v_\diml}$ and 
an oriented $\diml$-simplex $\vec{\sG}$ 
as $\vec{\sG}=[v_0,\ldots,v_\diml]$, 
where $v_0,\ldots,v_\diml$ specify the ordering of the spanning vertices.
\end{Remark}

\paragraph{Filtration.}
A {\it filtration} $\Fcal$ of a simplicial complex $K$ is a filtered sequence of subcomplexes of $K$, 
$\Fcal:\emptyset=K_0\subseteq K_1\subseteq\ldots\subseteq K_n=K$,
such that $K_{i}$ and $K_{i-1}$ differ by one simplex denoted by $\sigma^\Fcal_{i}$. 
We let $i$ be the {\it index} of $\sG^\Fcal_{i}$ in $\Fcal$ and denote it as  $\ind(\sG^\Fcal_i)=i$. 
A subcomplex $K_i$ in the filtered sequence of $\Fcal$ is also referred to as a {\it partial complex}.



\paragraph{Simplicial homology.}
We provide a brief overview of simplicial homology used in this paper.
See any standard book on the topic, e.g.~\cite{munkres2018elements}.
Let $q\geq 0$, $K$ be a simplicial complex, and $\Gbb$ be an abelian group.
The {\it $q_\text{th}$ chain group} $\Chn_q(K;\Gbb)$ 
is defined to be the abelian group containing all finite sums of the form $\sum_i n_i\vec{\sG}_i$,
where $n_i\in \Gbb$ and $\vec{\sG}_i$ is an oriented $q$-simplex of $K$.
Each element in $\Chn_q(K;\Gbb)$ is called a {\it$q$-chain} of $K$.
Note that for two oppositely oriented $q$-simplices $\vec{\sG}$ and $\vec{\sG}'$,
we have that $n\vec{\sG}=(-n)\vec{\sG}'$ for any $n\in\Gbb$.
Therefore, $\Chn_q(K;\Gbb)$ can be interpreted as a direct sum of $N_q$ copies of $\Gbb$
where $N_q$ is the number of $q$-simplices of $K$ 
and each copy of $\Gbb$ corresponds to a $q$-simplex of $K$.
The {\it $q_\text{th}$ boundary operator} $\partial_q:\Chn_q(K;\Gbb)\to\Chn_{q-1}(K;\Gbb)$
is a group homomorphism  
such that for any oriented $q$-simplex $[v_0,\ldots,v_\diml]$
\[\partial_q\big([v_0,\ldots,v_q]\big)=\sum_{i=0}^q(-1)^i[v_0,\ldots,\wdhat{v}_i,\ldots,v_q]\]
where the notation $[v_0,\ldots,\wdhat{v}_i,\ldots,v_q]$ means that $\wdhat{v}_i$ is deleted from the simplex.
For brevity, 
we often omit the subscript of the boundary operator $\partial_q$
and denote it as $\partial$
when this does not cause any confusion.
The kernel of $\partial_q$ is called the {\it $q_\text{th}$ cycle group} of $K$
and is denoted as $\Zyc_q(K;\Gbb)$.
The image of $\partial_{q+1}$ is called the {\it $q_\text{th}$ boundary group} of $K$
and is denoted as $\Bnd_q(K;\Gbb)$.
A $q$-chain in $\Zyc_q(K;\Gbb)$ is called a {\it $q$-cycle}
and a $q$-chain in $\Bnd_q(K;\Gbb)$ is called a {\it $q$-boundary}.
For a $q$-chain $A$, the $(q-1)$-chain $\partial(A)$ is also called the {\it boundary} of $A$.

A fundamental fact in homology theory is that $\partial_{q}\partial_{q+1}=0$ for any $q$.
This implies that $\Bnd_q(K;\Gbb)\subseteq\Zyc_q(K;\Gbb)$.
The {\it $q_\text{th}$ homology group} of $K$ denoted by $\Hm_q(K;\Gbb)$
is defined as the quotient $\Zyc_q(K;\Gbb)/\Bnd_q(K;\Gbb)$.
Each coset in $\Hm_q(K;\Gbb)$ is called a {\it homology class}
and a cycle is said to be {\it homologous to} another cycle if they belong to the same homology class.
As any boundary cycle represents the homology class $0$ in $\Hm_q(K;\Gbb)$,
a boundary is also said to be {\it null-homologous}.

The abelian group $\Gbb$ in the above definitions is called the {\it coefficient group}
for the homology groups.
Sometimes,
when the coefficient group $\Gbb$ is clear,
we simply drop it and denote a chain group as $\Chn_q(K)$.
This applies to other groups defined in simplicial homology.
In this paper, two coefficient groups $\Zbb_2$ and $\Zbb$ are used for simplicial homology.
When not explicitly stated,
the coefficients are assumed to be in $\Zbb_2$.
With $\Zbb_2$ coefficients, 
the orientations of simplices no longer matter
and a $q$-chain can be interpreted as a set of $q$-simplices
with summation of two $q$-chains being the {\it symmetric difference}.
A $q$-cycle is then a set of $q$-simplices
where every $(q-1)$-face of these simplices adjoins an even number of
$q$-simplices.
Also note that because $\Zbb_2$ is a field, 
all groups defined in simplicial homology with $\Zbb_2$ coefficients
become vector spaces 
and homomorphisms between these groups (such as $\partial$) become linear maps.


\begin{Definition}[$\diml$-weighted] \label{dfn:q-weighted}
A simplicial complex $K$ is {\it $\diml$-weighted} if each $\diml$-simplex $\sG$ of $K$ has a non-negative finite weight $w(\sG)$.
The weight of a $\diml$-chain $A$ of $K$ is then defined as $w(A)=\sum_{\sG\in A}w(\sG)$.
\end{Definition}

\begin{Definition}[$\diml$-connected] \label{dfn:q-conn}
Let $K$ be a simplicial complex, for $\diml\geq 1$, 
two $\diml$-simplices $\sG$ and $\sG'$ of $K$ are {\it $\diml$-connected in $K$} if there is a sequence of $\diml$-simplices of $K$, $(\sG_0,\ldots,\sG_l)$, 
such that $\sG_0=\sG$, $\sG_l=\sG'$, and for all $0\leq i<l$, $\sG_i$ and $\sG_{i+1}$ share a $(\diml-1)$-face. 
The property of $\diml$-connectedness defines 
an equivalence relation on $\diml$-simplices of $K$.
Each set in the partition induced by the equivalence relation 
constitutes a {\it $\diml$-connected component} of $K$.
We say $K$ is {\it $\diml$-connected} if any two $\diml$-simplices of $K$ are $\diml$-connected in $K$.
\end{Definition}
\begin{Remark}
See Figure~\ref{fig:dual-graph} for an example of 1-connected components and 2-connected components.
\end{Remark}

\begin{Definition}[$\diml$-connected cycle]
A $\diml$-cycle $\zG$ (with $\Zbb_2$ coefficients) is {\it $\diml$-connected} 
if the complex derived by taking the closure of the simplicial set $\zG$ is $\diml$-connected.
\end{Definition}

\paragraph{Persistent homology.}
We will provide a brief description of persistent homology.
We recommend the book by Edelsbrunner and Harer~\cite{edelsbrunner2010computational} for a detailed explanation of this topic
and the book by Chazal~et~al.~\cite{chazal2016structure} for its underlying Mathematical structure, {\it persistence module}.
Note that persistent homology in this paper is always assumed to be with $\Zbb_2$ coefficients.
The persistence algorithm starts with a filtration $\Fcal:\emptyset=K_0\subseteq K_1\subseteq\ldots\subseteq K_n=K$
of a simplicial complex $K$,
and for each simplex $\sG_i^\Fcal$,
inspects whether $\partial(\sG_i^\Fcal)$ is a boundary in $K_{i-1}$.
If $\partial(\sG_i^\Fcal)$ is a boundary in $K_{i-1}$, $\sG_i^\Fcal$ is called {\it positive};
otherwise, it is called {\it negative}.
The $d$-chains (or $d$-cycles) in $K_i$ that are not in $K_{i-1}$ are said to be {\it born in} $K_i$
or {\it created by} $\sG_i^\Fcal$.
A positive $d$-simplex creates some $d$-cycles
and a negative $d$-simplex makes some $(d-1)$-cycles become boundaries.
In the latter case, we also say that the negative $d$-simplex {\it kills} or {\it destroys} those $(d-1)$-cycles.
What is central to the persistence algorithm is a notion called {\it pairing}:
A positive simplex is initially {\it unpaired} when introduced;
when a negative $d$-simplex $\sG_i^\Fcal$ comes, 
the algorithm finds a $(d-1)$-cycle
created by an unpaired positive $(d-1)$-simplex $\sG_j^\Fcal$ 
which is homologous to $\partial(\sG_i^\Fcal)$
and {\it pair} $\sG_j^\Fcal$ {\it with} $\sG_i^\Fcal$.
Alongside the pairing, 
a {\it finite interval} $[j,i)$ is added to the $(d-1)_\text{th}$ {\it persistence diagram},
which is denoted by $\persdgm_{d-1}(\Fcal)$.
After all simplices are processed, some positive simplices may still be unpaired.
For each $\sG_i^\Fcal$ of these unpaired simplices,
an {\it infinite interval} $[i,+\infty)$ is added to $\persdgm_d(\Fcal)$,
where $d$ is the dimension of $\sG_i^\Fcal$.

Note that the pairing in the persistence algorithm 
for a given filtration is unique.
Also note that in this paper,
we assume a filtration of a complex is given
and the persistence intervals start and end with indices of the paired simplices.
However, in real-life applications,
one is often given a function on a simplicial complex.
To produce the persistence intervals,
a filtration needs to be derived
and the endpoints of the intervals 
are taken as function values on the paired simplices.
In such cases, 
we can associate a given interval to its simplex pair,
take the indices of the paired simplices,
and get an interval which can serve as an input to our algorithms.

\paragraph{The persistent cycle problems.}
We can now formally define the minimal persistent cycle problems:

\begin{problem}[PCYC-FIN$_\Dim$] \label{prob:pcyc-fin}
Given a finite $\Dim$-weighted simplicial complex $K$, 
a filtration $\Fcal: \emptyset=K_0\subseteq K_1\subseteq\ldots\subseteq K_n=K$, 
and a finite interval $[\birth,\death)\in\persdgm_\Dim(\Fcal)$, 
this problem asks for computing a $\Dim$-cycle with the minimal weight which is born in $K_{\birth}$ and becomes a boundary in $K_{\death}$.
\end{problem}

\begin{problem}[PCYC-INF$_\Dim$] \label{prob:pcyc-inf}
Given a finite $\Dim$-weighted simplicial complex $K$, 
a filtration $\Fcal: \emptyset=K_0\subseteq K_1\subseteq\ldots\subseteq K_n=K$, 
and an infinite interval $[\birth,+\infty)\in\persdgm_\Dim(\Fcal)$,
this problem asks for computing a $\Dim$-cycle with the minimal weight which is born in $K_\birth$.
\end{problem}

\begin{Remark}
The definitions of the above two problems are derived directly from the definition of persistent $d$-cycles~\cite{dey19pers1cyc}.
\end{Remark}

\paragraph{Undirected flow network.}
An {\it undirected flow network} $(G,s_1,s_2)$ consists of an undirected graph $G$ 
with vertex set $V(G)$ and edge set $E(G)$, 
a capacity function $c:E(G)\to [0,+\infty]$, 
and two non-empty disjoint subsets $s_1$ and $s_2$ of $V(G)$. 
Vertices in $s_1$ are referred to as {\it sources}
and vertices in $s_2$ are referred to as {\it sinks}.
A {\it cut} $(S,T)$ of $(G,s_1,s_2)$ consists of two disjoint subsets $S$ and $T$ of $V(G)$ such that $S\cup T=V(G)$, $s_1\subseteq S$, and $s_2\subseteq T$. 
The set of edges 
that connect a vertex in $S$ and a vertex in $T$ are referred as the edges
{\it across} the cut $(S,T)$ and is denoted as $\xi(S,T)$.
The {\it capacity} of a cut $(S,T)$ is defined as $c(S,T)=\sum_{e\in\xi(S,T)}c(e)$. 
A {\it minimal cut} of $(G,s_1,s_2)$ is a cut with the minimal capacity.
Note that we allow parallel edges in $G$ (see Figure~\ref{fig:dual-graph}) 
to ease the presentation.
These parallel edges can be merged into one edge during computation.

\section{Minimal persistent $\Dimless$-cycles of finite intervals for weak $(\Dimtop)$-pseudomanifolds}
\label{sec:fin-alg}

\begin{figure}
  \centering
  \subfloat[]{\includegraphics[width=0.23\linewidth]{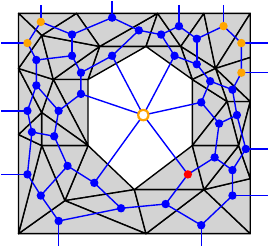}\label{fig:fin-ex1}}
  \hspace{0.5em}
  \subfloat[]{
    \raisebox{0.007\linewidth}{
      \includegraphics[width=0.202\linewidth]{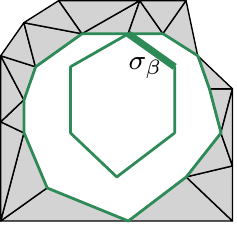}
    }\label{fig:fin-ex2}}
  \hspace{0.5em}
  \subfloat[]{
    \raisebox{0.007\linewidth}{
      \includegraphics[width=0.202\linewidth]{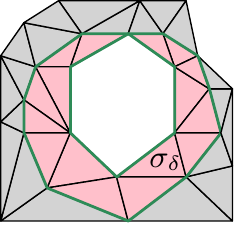}
    }\label{fig:fin-ex3}}
  \hspace{0.5em}
  \subfloat[]{\includegraphics[width=0.23\linewidth]{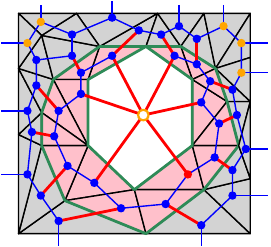}\label{fig:fin-ex4}}
  \caption{
    An example of the constructions in our algorithm 
    showing the duality between persistent cycles and cuts having finite capacity for $\Dimless=1$.
    (a)~The input weak 2-pseudomanifold $K$ with its dual flow network drawn in blue, 
    where the central hollow vertex denotes the dummy vertex,
    the red vertex denotes the source,
    and all the orange vertices (including the dummy one) denote the sinks.
    All ``dangled'' graph edges dual to the outer boundary 1-simplices actually 
    connect to the dummy vertex and these connections are not drawn.
    (b)~The partial complex $K_\birth$ in the input filtration $\Fcal$, 
    where the bold green 1-simplex denotes $\sG_\birth^\Fcal$ which
    creates the green 1-cycle.
    (c)~The partial complex $K_\death$ in $\Fcal$, where the 2-simplex $\sG_\death^\Fcal$
    creates the pink 2-chain killing the green 1-cycle.
    (d)~The green persistent 1-cycle of the interval $[\birth,\death)$ 
    is dual to a cut $(S,T)$ having finite capacity,
    where $S$ contains all the vertices inside the pink 2-chain and $T$ contains all the other vertices.
    The red graph edges denote those edges across $(S,T)$
    and their dual 1-chain is the green persistent 1-cycle.}
  \label{fig:fin-ex}
\end{figure}

In this section, we present an algorithm which computes minimal 
persistent $\Dimless$-cycles for finite intervals 
given a filtration of a weak $(\Dimtop)$-pseudomanifold when $\Dimless\geq 1$. 
The general process is as follows:
Suppose that the input weak $(\Dimtop)$-pseudomanifold is $K$ associated with a filtration $\Fcal:K_0\subseteq K_1\subseteq\ldots\subseteq K_n$ and
the task is to compute the minimal persistent cycle of 
a finite interval $[\birth,\death)\in\persdgm_\Dimless(\Fcal)$.
We first construct an undirected dual graph $G$ for $K$
where vertices of $G$ are dual to $(\Dimtop)$-simplices of $K$
and edges of $G$ are dual to $\Dimless$-simplices of $K$.
One dummy vertex termed as {\it infinite vertex}
which does not correspond to any $(\Dimtop)$-simplices is added to
$G$ for graph edges dual to those boundary $\Dimless$-simplices.
We then build an undirected flow network on top of $G$ 
where the source is the vertex dual to $\sG_\death^\Fcal$
and the sink is the infinite vertex along with
the set of vertices dual to those $(\Dimtop)$-simplices which are added to $\Fcal$ after $\sG_\death^\Fcal$.
If a $\Dimless$-simplex is $\sG_\birth^\Fcal$ or added to $\Fcal$ before $\sG_\birth^\Fcal$,
we let the capacity of its dual graph edge be its weight;
otherwise, we let the capacity of its dual graph edge be $+\infty$.
Finally, we calculate a minimal cut of this flow network 
and return the $\Dimless$-chain dual to the edges across the minimal cut as a minimal persistent cycle of the interval.

The intuition of the above algorithm is best explained by an example in Figure~\ref{fig:fin-ex}, where $d=1$.
The key to the algorithm is the duality between persistent cycles of the input interval 
and cuts of the dual flow network having finite capacity.
To see this duality,
first consider a persistent $\Dimless$-cycle $\zG$ of the input interval $[\birth,\death)$.
There exists a $(\Dimtop)$-chain $A$ in $K_\death$ created by $\sG_\death^\Fcal$
whose boundary equals $\zG$,
making $\zG$ killed.
We can let $S$ be the set of graph vertices dual to the simplices in $A$
and let $T$ be the set of the remaining graph vertices,
then $(S,T)$ is a cut.
Furthermore, $(S,T)$ must have finite capacity as 
the edges across it are exactly dual to the $\Dimless$-simplices in $\zG$ 
and the $\Dimless$-simplices in $\zG$ have indices in $\Fcal$ less than or equal to $\birth$.
On the other hand, 
let $(S,T)$ be a cut with finite capacity,
then the $(\Dimtop)$-chain whose simplices are dual to the vertices in $S$
is created by $\sG_\death^\Fcal$.
Taking the boundary of this $(\Dimtop)$-chain, we get a $\Dimless$-cycle $\zG$.
Because $\Dimless$-simplices of $\zG$ are exactly dual to the edges across $(S,T)$
and each edge across $(S,T)$ has finite capacity,
$\zG$ must reside in $K_\birth$. 
We only need to ensure that
$\zG$ contains $\sG_\birth^\Fcal$ in order to show that $\zG$ is a persistent cycle of $[\birth,\death)$.
In Section~\ref{sec:corr-fin}, we argue that
$\zG$ actually contains $\sG_\birth^\Fcal$, so $\zG$ is indeed a persistent cycle.
Note that while the above explanation introduces the general idea,
the rigorous statement and proof of the duality are articulated by Proposition~\ref{prop:algr-fin-bar-fincut-pers} and~\ref{prop:algr-fin-bar-pers-fincut}.



\begin{algorithm}[h]
\caption{Computing minimal persistent $\Dimless$-cycles of finite intervals for weak $(\Dimtop)$-pseudomanifolds}

\vspace{1ex}

\textbf{Input:}

\hspace{7pt}$K$: finite $\Dimless$-weighted weak $(\Dimtop)$-pseudomanifold

\hspace{7pt}$\Dimless$: integer $\geq 1$

\hspace{7pt}$\Fcal$: filtration $K_0\subseteq K_1\subseteq\ldots\subseteq K_n$ of $K$

\hspace{7pt}$[\birth,\death)$: finite interval of $\persdgm_{\Dimless}(\Fcal)$

\textbf{Output:} 

\hspace{7pt}minimal persistent $\Dimless$-cycle of $[\birth,\death)$

\vspace{1ex}

\begin{algorithmic}[1]
\Procedure{MinPersCycFin}{$K,\Dimless,\Fcal,[\birth,\death)$}
\LineComment{\textsf{set up the complex $\wdtild{K}$ being worked on}}
\State $C^\Dimtop\leftarrow$ $(\Dimtop)$-connected component of $K$ containing $\sG^\Fcal_\death$
\label{alg-line:set-Cd}
\State $\wdtild{K}\leftarrow$ closure of the simplicial set $C^\Dimtop$ \label{alg-line:set-Ktild}
\LineComment{\textsf{construct dual graph}}
\State $(G,\thG)\leftarrow\textsc{DualGraphFin}(\wdtild{K},\Dimless)$ \label{alg-line:dual-graph-fin}
\LineComment{\textsf{assign capacity to $G$}}
\ForEach{$e\in E(G)$} \label{alg-line:assgn-cap-for-fin}
\If{$\ind(\thG\inv(e))\leq \birth$}
\State $c(e)\leftarrow w(\thG\inv(e))$
\Else
\State $c(e)\leftarrow +\infty$
\EndIf
\EndFor
\vspace{-0.5em}
\LineComment{\textsf{set the source}}
\State $s_1\leftarrow\Set{\thG(\sigma^\Fcal_\death)}$
\LineComment{\textsf{set the sink}}
\State $s_2\leftarrow\Set{v\in V(G)\given v\neq \phi,\,\ind(\thG\inv(v))>\death}$
\If{$\phi\in V(G)$}
\State $s_2\leftarrow s_2\union\Set{\phi}$ \label{alg-line:sink-loop}
\EndIf
\State $(S^*,T^*)\leftarrow$ min-cut of $(G,s_1,s_2)$ \label{alg-line:min-cut-fin}
\State \Return $\thG\inv(\xi(S^*,T^*))$ \label{alg-line:return-fin}
\EndProcedure
\end{algorithmic}
\label{algr:min-cyc-fin-bar}
\end{algorithm}

We list the pseudo-code in Algorithm~\ref{algr:min-cyc-fin-bar} and 
it works as follows:
Line~\ref{alg-line:set-Cd} and~\ref{alg-line:set-Ktild} set up a complex $\wdtild{K}$ that the algorithm mainly works on,
where $\wdtild{K}$ is taken as the closure of the $(\Dimtop)$-connected component 
of $K$ containing $\sG_\death^\Fcal$.
The reason for working on $\wdtild{K}$ instead of the entire complex is explained later in this section.
Line~\ref{alg-line:dual-graph-fin} constructs the dual graph $G$ from $\wdtild{K}$ and
line~\ref{alg-line:assgn-cap-for-fin}$-$\ref{alg-line:sink-loop} 
builds the flow network on top of $G$.
Note that we denote the infinite vertex by $\phi$.
Line~\ref{alg-line:min-cut-fin}
computes a minimal cut for the flow network and
line~\ref{alg-line:return-fin} returns the $\Dimless$-chain dual to the edges across the minimal cut.
In the pseudo-codes of this paper, to ease the exposition,
we treat a Mathematical function as a computer program object.
For example,
the function $\thG$ returned by $\textsc{DualGraphFin}$ in Algorithm~\ref{algr:min-cyc-fin-bar}
denotes the bijection between the simplices of $\wdtild{K}$ 
and their dual vertices or edges
(see Section~\ref{sec:grph-cons-fin} for details).
In practice, these constructs can be easily implemented in any computer programming language.

To see the reason why we work on $\wdtild{K}$,
we first note that the dual graph constructed directly from $K$ may be disconnected\footnote{
  For an example in $d=1$, take $K$ as two disconnected triangulated 2-spheres.
  Its dual graph consists of two connected components.}.
While cuts are still well-defined for a disconnected flow network,
one may prefer a connected one as 
the minimal cut computation only concerns the graph component containing the source.
By constructing the dual graph from $\wdtild{K}$,
it can be ensured that the graph is connected.
In order for Algorithm~\ref{algr:min-cyc-fin-bar} to work,
one has to further show that the sink is non-empty 
so that the computed persistent cycle is non-empty. 
This is verified in Proposition~\ref{prop:algr-fin-bar-nonempt}.
An intuitive reason why the computation from $\wdtild{K}$ is still correct
is as follows:
Each persistent $\Dimless$-cycle $\zG$ of the given interval corresponds to a
$(\Dimtop)$-chain $A$ which kills $\zG$, i.e., $\partial(A)=\zG$.
Suppose that $A$ is not entirely contained in $\wdtild{K}$. Notice that
$A\intersect \wdtild{K}\neq \emptyset$ and contains at least the killer simplex $\sG_\death^\Fcal$. 
Then $\partial(A\intersect \wdtild{K})$ 
must be a persistent cycle of the interval
residing in $\wdtild{K}$
which has a smaller weight.
Hence, a~minimal persistent cycle must reside in $\wdtild{K}$.
In Section~\ref{sec:corr-fin}, we formally verify the construction.

\paragraph{Complexity.} 
The time complexity of Algorithm~\ref{algr:min-cyc-fin-bar} depends on 
the encoding scheme of the input and 
the data structure used for representing a simplicial complex.
For encodings of the input,
we assume $K$ and $\Fcal$ to be represented by a sequence of all the simplices of $K$
ordered by their indices in $\Fcal$,
where each simplex is denoted by its set of vertices.
We also assume a simple yet reasonable simplicial complex data structure as follows:
In each dimension, simplices are mapped to integral identifiers ranging from 0 to 
the number of simplices in that dimension minus 1;
each $\diml$-simplex has an array (or linked list) storing all the id's of its $(\diml+1)$-cofaces;
a hash map for each dimension is maintained for the query of the integral id 
of each simplex in that dimension based on the spanning vertices of the simplex.
We further assume $\Dimless$ to be constant.
By the above assumptions, 
let $n$ be the size (number of bits) of the encoded input,
then there are no more than $n$ elementary $O(1)$ operations 
in line~\ref{alg-line:set-Cd} and~\ref{alg-line:set-Ktild}.
So, the time complexity of line~\ref{alg-line:set-Cd} and~\ref{alg-line:set-Ktild}
is $O(n)$.
It is not hard to verify that the flow network construction also takes $O(n)$ time
so the time complexity of Algorithm~\ref{algr:min-cyc-fin-bar} is determined by the minimal cut algorithm.
Using the max-flow algorithm by Orlin~\cite{orlin2013max},
the time complexity of Algorithm~\ref{algr:min-cyc-fin-bar} becomes $O(n^2)$.

\paragraph{}
In the rest of this section, 
we first explain the bijection $\thG$ returned by \textsc{DualGraphFin},
then prove the correctness of the algorithm.


\subsection{The bijection \texorpdfstring{$\bm{\thG}$}{theta}}
\label{sec:grph-cons-fin}


The vertex set $V(G)$ of $G$ contains vertices 
which correspond to the $(\Dimtop)$-simplices of $\wdtild{K}$.
The set $V(G)$ may also contain an infinite vertex $\phi$ 
if $\wdtild{K}$ contains any boundary $\Dimless$-simplex.
We define a bijection 
\[\thG:\Set{(\Dimtop)\text{-simplices of }\wdtild{K}}\to V(G)\smallsetminus \Set{\phi}\]
such that for any $(\Dimtop)$-simplex $\sG^\Dimtop$ of $\wdtild{K}$,
$\thG(\sG^\Dimtop)$ is the vertex that $\sG^\Dimtop$ is dual to.
Similarly, we define another bijection
\[\thG:\Set{\Dimless\text{-simplices of }\wdtild{K}}\to E(G)\]
using the same notation $\thG$.

Note that we
can take the image of a subset of the domain under a function. 
Therefore, if $(S,T)$ is a cut for a flow network built on $G$, 
then $\thG\inv(\xi(S,T))$ denotes the set of $\Dimless$-simplices dual to the edges across the cut. 
Also note that since simplicial chains with $\Zbb_2$ coefficients can be interpreted as sets, $\thG\inv(\xi(S,T))$ is also a $\Dimless$-chain.

\subsection{Algorithm correctness} \label{sec:corr-fin}
In this subsection, we prove the correctness of Algorithm~\ref{algr:min-cyc-fin-bar}. 
Some of the symbols we use refer to Algorithm~\ref{algr:min-cyc-fin-bar}.

\begin{proposition}
\label{prop:algr-fin-bar-nonempt}
In Algorithm~\ref{algr:min-cyc-fin-bar}, the sink $s_2$ is not an empty set.
\end{proposition}
\begin{proof}
For contradiction, suppose that $s_2$ is an empty set. Then, $\phi\not\in V(G)$ 
and $\sG_\death^\Fcal$ is the $(\Dimtop)$-simplex of $\wdtild{K}$ with the greatest index in $\Fcal$.
Because $\phi\not\in V(G)$, any $\Dimless$-simplex of $\wdtild{K}$ must be a face of two $(\Dimtop)$-simplices of $\wdtild{K}$, 
so the set of $(\Dimtop)$-simplices of $\wdtild{K}$ forms a $(\Dimtop)$-cycle created by ${\sG_\death^\Fcal}$. 
Then ${\sG_\death^\Fcal}$ must be a positive simplex in $\Fcal$, which is a contradiction.
\end{proof}

The following two propositions specify the duality mentioned at the beginning of this section:

\begin{proposition}
\label{prop:algr-fin-bar-fincut-pers}
For any cut $(S,T)$ of $(G,s_1,s_2)$ with finite capacity, 
the $\Dimless$-chain $\zG=\thG\inv(\xi({S},{T}))$ is a persistent $\Dimless$-cycle of $[\birth,\death)$ 
and $w(\zG)=c(S,T)$.
\end{proposition}

\begin{proof}
Let $A=\thG\inv(S)$, we first want to prove $\zG=\partial(A)$,
so that $\zG$ is a cycle. 
Let $\sigma^{\Dimless}$ be any $\Dimless$-simplex of $\zG$, then $\thG(\sigma^{\Dimless})$ connects a vertex $u\in S$ and a vertex $v\in T$.
If $v=\phi$, 
then $\sigma^{\Dimless}$ cannot be a face of another $(\Dimtop)$-simplex in $K$ other than $\thG\inv(u)$.
So, $\sigma^{\Dimless}$ is a face of exactly one $(\Dimtop)$-simplex of $A$.
If $v\neq \phi$, then $\sigma^{\Dimless}$ is also a face of exactly one $(\Dimtop)$-simplex of $A$.  
Therefore, $\sigma^{\Dimless}\in\partial(A)$. 
On the other hand, let $\sigma^{\Dimless}$ be any $\Dimless$-simplex of $\partial(A)$, 
then $\sigma^{\Dimless}$ is a face of exactly one $(\Dimtop)$-simplex $\sG_0^\Dimtop$ of $A$.
If $\sigma^{\Dimless}$ is a face of another $(\Dimtop)$-simplex $\sG_1^\Dimtop$ in $K$, 
then $\sG^\Dimtop_1\in \wdtild{K}$ and $\sG^\Dimtop_1\not\in A$. 
So, $\thG(\sG^{\Dimless})$ connects the vertex $\thG(\sG^\Dimtop_0)\in S$ and the vertex $\thG(\sG^\Dimtop_1)\in T$ in the graph $G$.
If $\sigma^{\Dimless}$ is a face of exactly one $(\Dimtop)$-simplex in $K$, 
$\thG(\sG^{\Dimless})$ must connect $\thG(\sG^\Dimtop_0)\in S$ and $\phi\in T$ in $G$.
So we have $\thG(\sG^{\Dimless})\in\xi(S,T)$, i.e., $\sG^{\Dimless}\in\thG\inv(\xi(S,T))$.

We then show that $\zG$ is created by $\sG_\birth^\Fcal$. 
By Proposition~\ref{prop:algr-fin-bar-nonempt}, $\zG$ cannot be empty.
Therefore, for contradiction, we can
suppose that $\zG$ is created by a $\Dimless$-simplex $\sG^{\Dimless}\neq \sG_\birth^\Fcal$.
Because $c(S,T)$ has finite capacity,
we have that $\ind(\sG^{\Dimless})<\birth$.
We can let $\zG'$ be a persistent cycle of $[\birth,\death)$ and $\zG'=\partial(A')$ where $A'$ is a $(\Dimtop)$-chain of $K_\death$. Then we have $\zG+\zG'=\partial(A+A')$. 
Since $A$ and $A'$ are both created by $\sG_\death^\Fcal$, then $A+A'$ is created by a $(\Dimtop)$-simplex with an index less than $\death$ in $\Fcal$. 
So $\zG+\zG'$ is a $\Dimless$-cycle created by $\sG_\birth^\Fcal$ 
which becomes a boundary before $\sG_\death^\Fcal$ is added. 
This means that 
$\sG_\birth^\Fcal$ is already paired when $\sG_\death^\Fcal$ is added, 
contradicting the fact that $\sG_\birth^\Fcal$ is paired with $\sG_\death^\Fcal$.
Similarly, we can prove that $\zG$ is not a boundary until $\sG_\death^\Fcal$ is added, so $\zG$ is a persistent cycle of $[\birth,\death)$. Since $(S,T)$ has finite capacity, we must have
\[
c(S,T)=\sum_{e\in \thG(\zG)}c(e)=\sum_{\thG\inv(e)\in \zG}w(\thG\inv(e))=w(\zG)
\qedhere
\]
\end{proof}

\begin{proposition}
\label{prop:algr-fin-bar-pers-fincut}
For any persistent $\Dimless$-cycle $\zG$ of $[\birth,\death)$, 
there exists a cut $(S,T)$ of $(G,s_1,s_2)$ such that $c(S,T)\leq w(\zG)$.
\end{proposition}
\begin{proof}
Let $A$ be a $(\Dimtop)$-chain in $K_\death$ such that $\zG=\partial(A)$. 
Note that $A$ is created by $\sigma_\death^\Fcal$ 
and $\zG$ is the set of $\Dimless$-simplices 
which are face of exactly one $(\Dimtop)$-simplex of $A$. 
Let $\zG'=\zG\intersect\wdtild{K}$ and
$A'=A\intersect\wdtild{K}$,
we claim that $\zG'=\partial(A')$.
To prove this, 
first let $\sG^{\Dimless}$ be any $\Dimless$-simplex of $\zG'$, 
then $\sG^{\Dimless}$ is a face of exactly one $(\Dimtop)$-simplex $\sG^\Dimtop$ of $A$.
Since $\sG^{\Dimless}\in \wdtild{K}$, 
it is also true that $\sG^\Dimtop\in \wdtild{K}$, 
so $\sG^\Dimtop\in A'$.
Then $\sG^{\Dimless}$ is a face of exactly one $(\Dimtop)$-simplex of $A'$, so $\sG^{\Dimless}\in\partial(A')$.
On the other hand, 
let $\sG^{\Dimless}$ be any $\Dimless$-simplex of $\partial(A')$, 
then $\sG^{\Dimless}$ is a face of exactly one $(\Dimtop)$-simplex $\sG_0^\Dimtop$ of $A'$.
Note that $\sG_0^\Dimtop\in A$ 
and we then want to prove that $\sG^{\Dimless}$ is a face of exactly one $(\Dimtop)$-simplex $\sG_0^\Dimtop$ of $A$.
Suppose that $\sG^{\Dimless}$ is a face of another $(\Dimtop)$-simplex $\sG^\Dimtop_1$ of $A$, 
then $\sG^\Dimtop_1\in \wdtild{K}$ because $\sG_0^\Dimtop\in \wdtild{K}$. 
So we have $\sG^\Dimtop_1\in A\intersect\wdtild{K}=A'$, 
contradicting the fact that $\sG^{\Dimless}$ is a face of exactly one $(\Dimtop)$-simplex of $A'$.
Then we have $\sG^{\Dimless}\in\partial(A)$. 
Since $\sG_0^\Dimtop\in \wdtild{K}$, we have $\sG^{\Dimless}\in \wdtild{K}$, 
which means that $\sG^{\Dimless}\in\zG'$.

Let $S=\thG(A')$ and $T=V(G)\smallsetminus S$, 
then it is true that $(S,T)$ is a cut of $(G,s_1,s_2)$ 
because $A'$ is created by $\sG_\death^\Fcal$.
We claim that $\thG\inv(\xi(S,T))=\partial(A')$.
The proof of the equality is similar to the one in the proof of Proposition~\ref{prop:algr-fin-bar-fincut-pers}.
It follows that $\xi(S,T)=\thG(\zG')$.
We then have that
\[c(S,T)=\sum_{e\in \thG(\zG')}c(e)=\sum_{\thG\inv(e)\in \zG'}w(\thG\inv(e))=w(\zG')\]
because each $\Dimless$-simplex of $\zG'$ has an index less than or equal to $\birth$ in $\Fcal$.

Finally, because $\zG'$ is a subchain of $\zG$, we must have $c(S,T)=w(\zG')\leq w(\zG)$.
\end{proof}

Combining the above facts, we can conclude:

\begin{theorem}
Algorithm~\ref{algr:min-cyc-fin-bar} computes a minimal persistent $\Dimless$-cycle 
for the given interval $[\birth,\death)$.
\end{theorem}
\begin{proof}
First, the flow network $(G,s_1,s_2)$ constructed 
by Algorithm~\ref{algr:min-cyc-fin-bar} must be valid by Proposition~\ref{prop:algr-fin-bar-nonempt}.
Next, because the interval $[\birth,\death)$ must have a persistent cycle, 
by Proposition~\ref{prop:algr-fin-bar-pers-fincut}, 
the flow network $(G,s_1,s_2)$ has a cut with finite capacity.
This means that $c(S^*,T^*)$ is finite.
By Proposition~\ref{prop:algr-fin-bar-fincut-pers},
the chain $\zG^*=\thG\inv(\xi(S^*,T^*))$ is a persistent cycle of $[\birth,\death)$.
Assume that $\zG^*$ is not a minimal persistent cycle of $[\birth,\death)$ and 
instead let $\zG'$ be a minimal persistent cycle of $[\birth,\death)$.
Then there exists a cut $(S',T')$ such that $c(S',T')\leq w(\zG')<w(\zG^*)=c(S^*,T^*)$ 
by Proposition~\ref{prop:algr-fin-bar-fincut-pers} and~\ref{prop:algr-fin-bar-pers-fincut}, 
contradicting the fact that $(S^*,T^*)$ is a minimal cut.
\end{proof}

\section{Minimal persistent $\Dimless$-cycles of infinite intervals for weak $(\Dimtop)$-pseudomanifolds embedded in $\real^\Dimtop$}
\label{sec:inf-alg}
We already mentioned that computing minimal persistent $\Dimless$-cycles ($\Dimless\geq 2$) 
for infinite intervals is NP-hard even if we restrict to weak $(d+1)$-pseudomanifolds
(see Section~\ref{sec:inf-hard} for a proof). 
However, when the complex is embedded in $\real^\Dimtop$, the problem becomes polynomially tractable.
In this section, we present an algorithm for this problem
in $\Dimless\geq 1$\footnote{
  As mentioned earlier, when $\Dimless=1$, this problem is polynomially tractable for arbitrary complexes.}.
The algorithm uses a similar duality described in Section~\ref{sec:fin-alg}.
However, a direct use of the approach in Section~\ref{sec:fin-alg} does not work.
For example, in Figure~\ref{fig:dual-graph},
1-simplices that
do not have any 2-cofaces cannot reside in any $2$-connected component of the given complex.
Hence, no cut in the flow network may correspond to a persistent cycle of the infinite interval 
created by such a $1$-simplex. Furthermore, unlike the finite interval case, we do not have
a negative simplex whose dual can act as a source in the flow network.

\begin{figure}
\centering
  \subfloat[]{\includegraphics[width=0.35\linewidth]{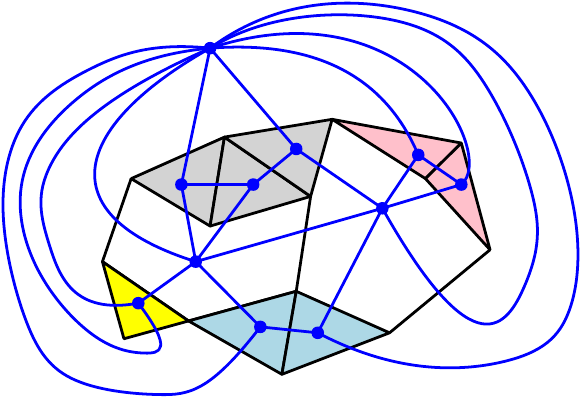}\label{fig:dual-graph}}
  \hspace{3em}
  \subfloat[]{\includegraphics[width=0.16\linewidth]{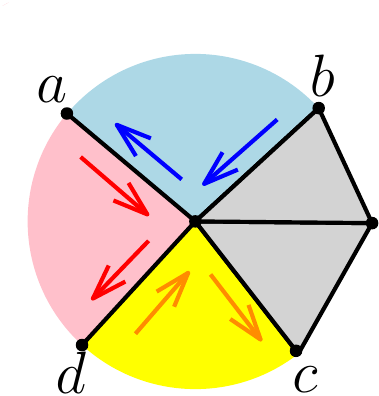}\label{fig:edge-pair}}

   \caption{
   (a) A weak 2-pseudomanifold $\wdtild{K}$ embedded in $\real^2$ with three voids. 
   Its dual graph is drawn in blue. 
   The complex has one 1-connected component and four 2-connected components 
   with the 2-simplices in different 2-connected components colored differently.
   (b) An example illustrating the pairing of boundary $\Dimless$-simplices
   in the neighborhood of a $(\Dimlsls)$-simplex for $\Dimless=1$.
   The four boundary 1-simplices produce
   six oriented boundary 1-simplices
   and the paired oriented 1-simplices are colored the same.}
\end{figure}

Let $(K,\Fcal,[\birth,+\infty))$ be an input to the problem
where $K$ is a weak $(\Dimtop)$-pseudomanifold embedded in $\real^\Dimtop$,
$\Fcal:K_0\subseteq K_1\subseteq\ldots\subseteq K_n$ is a filtration of $K$,
and $[\birth,+\infty)$ is an infinite interval of $\persdgm_\Dimless(\Fcal)$.
By the definition of the problem,
the task boils down to computing a minimal $\Dimless$-cycle containing $\sG_\birth^\Fcal$ in $K_\birth$.
Note that $K_\birth$ is also a weak $(\Dimtop)$-pseudomanifold embedded in $\real^\Dimtop$.

Generically, assume $\wdtild{K}$ is an arbitrary weak $(\Dimtop)$-pseudomanifold embedded in $\real^\Dimtop$
and we want to compute a minimal $\Dimless$-cycle containing 
a $\Dimless$-simplex $\wdtild{\sG}$ for $\wdtild{K}$.
By the embedding assumption,
the connected components of $\real^\Dimtop\smallsetminus|\wdtild{K}|$ are well defined
and we call them the {\it voids} of $\real^\Dimtop\smallsetminus|\wdtild{K}|$. 
The complex $\wdtild{K}$ has a natural (undirected) dual graph structure 
as exemplified by Figure~\ref{fig:dual-graph} for $\Dimless=1$,
where the graph vertices are dual to the $(\Dimtop)$-simplices as well as the voids
and the graph edges are dual to the $\Dimless$-simplices. 
The duality between cycles and cuts is as follows:
Since the ambient space $\real^\Dimtop$ is contractible (homotopy equivalent to a point), 
every $\Dimless$-cycle in $\wdtild{K}$ is the boundary 
of a $(\Dimtop)$-dimensional region obtained by point-wise union of
certain $(\Dimtop)$-simplices and/or voids.
We can derive a cut\footnote{
  The cut here is defined on a graph without sources and sinks, 
  so the cut is simply a partition of the vertex set into two sets.
} of the dual graph 
by putting all vertices contained in the $(\Dimtop)$-dimensional region into one vertex set
and putting the rest into the other vertex set.
On the other hand, 
for every cut of the graph,
we can take the point-wise union of all the $(\Dimtop)$-simplices and voids
dual to the graph vertices in one set of the cut
and derive a $(\Dimtop)$-dimensional region.
The boundary of the derived $(\Dimtop)$-dimensional region  
is then a $\Dimless$-cycle in $\wdtild{K}$. 
We observe that by making the source and sink dual to 
the two $(\Dimtop)$-simplices or voids that $\wdtild{\sG}$ adjoins,
we can build a flow network where a minimal cut produces 
a minimal $\Dimless$-cycle in $\wdtild{K}$ containing $\wdtild{\sG}$.

The efficiency of the above algorithm is in part determined by 
the efficiency of the dual graph construction.
This step requires identifying the voids
that the boundary $\Dimless$-simplices are incident on.
A straightforward approach would be to first group the boundary $\Dimless$-simplices 
into $\Dimless$-cycles by local geometry, 
and then build the nesting structure of these $\Dimless$-cycles 
to correctly reconstruct the boundaries of the voids.
This approach has a quadratic worst-case complexity.
To make the void boundary reconstruction faster,
we assume that the simplicial complex being worked on is $\Dimless$-connected
so that building the nesting structure is not needed.
Our reconstruction then runs in almost linear time.
To satisfy the $\Dimless$-connected assumption,
we begin our algorithm by
taking $\wdtild{K}$ as a $\Dimless$-connected subcomplex of $K_\birth$ containing $\sG_\birth^\Fcal$
and continue only with this $\wdtild{K}$.
The computed output is still correct because
the minimal cycle in $\wdtild{K}$ is again a minimal cycle in $K_\birth$
as shown in Section~\ref{sec:corr-inf}.

\begin{algorithm}
\caption{Computing minimal persistent $\Dimless$-cycles of infinite intervals for weak $(\Dimtop)$-pseudomanifolds embedded in $\real^\Dimtop$}
\vspace{1ex}
\textbf{Input:} 

\hspace{7pt}$K$: finite $\Dimless$-weighted weak $(\Dimtop)$-pseudomanifold embedded in $\real^\Dimtop$

\hspace{7pt}{$\Dimless$: integer $\geq 1$}

\hspace{7pt}$\Fcal$: filtration $K_0\subseteq K_1\subseteq\ldots\subseteq K_n$ of $K$

\hspace{7pt}$[\birth,+\infty)$: infinite interval of $\persdgm_{\Dimless}(\Fcal)$

\textbf{Output:} 

\hspace{7pt}minimal persistent $\Dimless$-cycle of $[\birth,+\infty)$

\vspace{1ex}

\begin{algorithmic}[1]
\Procedure{MinPersCycInf}{$K,\Dimless,\Fcal,[\birth,+\infty)$}
\LineComment{\textsf{set up the complex $\wdtild{K}$ being worked on}}
\State $K'_\birth\leftarrow\textsc{Prune}(K_\birth,\Dimless)$ \label{alg-line:prune}
\State $C_\birth\leftarrow$ $\Dimless$-connected component of $K'_\birth$ containing $\sG_\birth^\Fcal$
\label{alg-line:conn-inf}
\State $\SG^\Dimtop\leftarrow \{\sG\in K'_\birth\,|\, \sG\text{ is a }(\Dimtop)\text{-simplex}$
$\text{and all }\Dimless\text{-faces of }\sG\text{ are in }C_\birth\}$

\label{alg-line:add-top-simp}
\State $\wdtild{K}\leftarrow($closure of the simplicial set $C_\birth)\union\SG^\Dimtop$ \label{alg-line:closure-inf}
\LineComment{\textsf{construct dual graph}}
\State $(\vec{\zG}_1,\ldots,\vec{\zG}_k)\leftarrow\textsc{VoidBoundary}(\wdtild{K},\Dimless)$ 
 \label{alg-line:void-bound}
\State $(G,\thG)\leftarrow\textsc{DualGraphInf}(\wdtild{K},\Dimless,\vec{\zG}_1,\ldots,\vec{\zG}_k)$\label{alg-line:dual-graph-inf}

\LineComment{\textsf{assign capacity to $G$}}
\ForEach{$e\in E(G)$} \label{alg-line:assgn-cap-for}
\State $c(e)\leftarrow w(\thG\inv(e))$
\EndFor

\State $(v_1,v_2)\leftarrow$ end vertices of edge $\thG(\sG_\birth^\Fcal)$ in $G$
\LineComment{\textsf{set the source}}
\State $s_1\leftarrow\Set{v_1}$
\LineComment{\textsf{set the sink}}
\State $s_2\leftarrow\Set{v_2}$ \label{alg-line:sink-inf}
\State $(S^*,T^*)\leftarrow$ min-cut of $(G,s_1,s_2)$ \label{alg-line:min-cut-inf}
\State \Return $\thG\inv(\xi(S^*,T^*))$\label{alg-line:return-inf}
\EndProcedure
\end{algorithmic}
\label{algr:min-cyc-inf-bar}
\end{algorithm}

We list the pseudo-code in Algorithm~\ref{algr:min-cyc-inf-bar}
and it works as follows:
Line~\ref{alg-line:prune}$-$\ref{alg-line:closure-inf} set up the complex $\wdtild{K}$ that the algorithm works on.
Line~\ref{alg-line:prune} prunes $K_\birth$ to produce a complex $K_\birth'$.
Given $(K_\birth,\Dimless)$, 
the \textsc{Prune} subroutine iteratively deletes a $\Dimless$-simplex $\sG^{\Dimless}$ of $K_\birth$
such that there is a $(\Dimless-1)$-face of $\sG^{\Dimless}$ having $\sG^{\Dimless}$ as the only  $\Dimless$-coface
(i.e., $\sG^{\Dimless}$ is a dangled $\Dimless$-simplex), 
until no such $\Dimless$-simplex can be found.
It is not hard to verify that \textsc{Prune} only deletes
$\Dimless$-simplices not residing in any $\Dimless$-cycles,
so a minimal $d$-cycle containing $\sG_\birth^\Fcal$ is never deleted.
We perform the pruning because 
it can reduce the graph size for the minimal cut computation 
which is more time consuming.
In line~\ref{alg-line:conn-inf}$-$\ref{alg-line:closure-inf},
we take the $\Dimless$-connected component $C_\birth$ of $K_\birth'$
containing $\sG_\birth^\Fcal$ and add a set $\SG^\Dimtop$ of 
$(\Dimtop)$-simplices to the closure of $C_\birth$ to form $\wdtild{K}$.
The set $\SG^\Dimtop$ contains all $(\Dimtop)$-simplices of $K_\birth'$ 
whose $\Dimless$-faces reside in~$C_\birth$.
The reason of adding the set $\SG^\Dimtop$ is
to reduce the number of voids for the complex $\wdtild{K}$
and in turn reduce the running time of the subsequent void boundary reconstruction.
For example, in Figure~\ref{fig:void_bd_recon2}, 
we could treat the entire complex as $K_\birth'$,
all 1-simplices as $C_\birth$,
and all 2-simplices as $\SG^\Dimtop$.
If we do not add $\SG^\Dimtop$ to the closure of $C_\birth$, 
there will be seven more voids corresponding to the seven 2-simplices.
Line~\ref{alg-line:void-bound} reconstructs the void boundaries for $\wdtild{K}$.
Each returned $\vec{\zG}_j$ denotes 
a set of $\Dimless$-simplices forming the boundary of a void.
As indicated in Section~\ref{sec:void-bound-recon},
the $\Dimless$-simplices in a void boundary are oriented.
Line~\ref{alg-line:dual-graph-inf} constructs the dual graph $G$ based on the reconstructed void boundaries.
Similar to Algorithm~\ref{algr:min-cyc-fin-bar},
the function $\thG$ returned by \textsc{DualGraphInf}
denotes the bijection from $\Dimless$-simplices of $\wdtild{K}$ to $E(G)$.
Line~\ref{alg-line:assgn-cap-for}$-$\ref{alg-line:sink-inf} 
build the flow network on top of $G$.
The capacity of each edge is equal to the weight of its dual $\Dimless$-simplex
and the source and sink are selected as previously described.
Line~\ref{alg-line:min-cut-inf} 
computes a minimal cut for the flow network
and line~\ref{alg-line:return-inf} returns the $\Dimless$-chain 
dual to the edges across the minimal cut.

\paragraph{Complexity.}
We make the same assumptions as in the complexity analysis for Algorithm~\ref{algr:min-cyc-fin-bar}.
Since the void boundary reconstruction needs to sort the $\Dimless$-cofaces of certain $(\Dimlsls)$-simplices,
its worst-case time complexity is $O(n\log n)$.
Then, all operations other than the minimal cut computation take $O(n\log n)$ time.
Therefore, similar to Algorithm~\ref{algr:min-cyc-fin-bar},
Algorithm~\ref{algr:min-cyc-inf-bar} achieves a complexity of $O(n^2)$
by using Orlin's max-flow algorithm~\cite{orlin2013max}.

\paragraph{}
In the rest of this section, 
we first describe the subroutine \textsc{VoidBoundary} 
invoked by Algorithm~\ref{algr:min-cyc-inf-bar} and then
prove the correctness of the algorithm.

\subsection{Void boundary reconstruction}\label{sec:void-bound-recon}

\begin{figure}
\centering
  \subfloat[]{\includegraphics[width=0.25\linewidth]{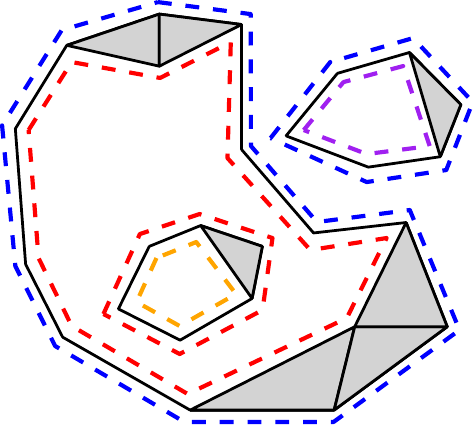}\label{fig:void_bd_recon1}}
  \hspace{6em}
  \subfloat[]{\includegraphics[width=0.25\linewidth]{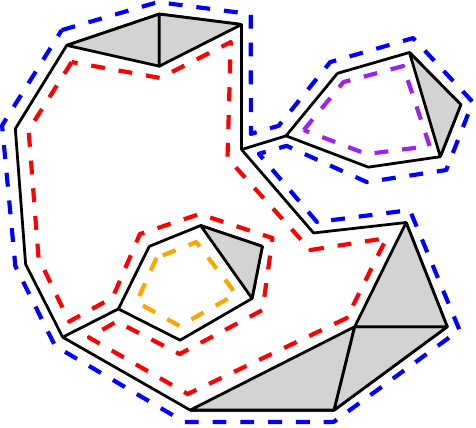}\label{fig:void_bd_recon2}}

   \caption{
   Examples showing how the void boundaries are reconstructed for $\Dimless=1$.
   (a) Oriented boundary $1$-simplices (drawn as dashed edges)
      of a simplicial complex
   are grouped into six 1-cycles
   and these six 1-cycles are further grouped into four void boundaries
   with each void boundary identically colored.
   (b) With the complex being 1-connected,
   the four grouped 1-cycles are exactly the boundaries of the four voids.
   }
   \label{fig:void_bd_recon}
\end{figure}

As previously stated, the object of the reconstruction is to identify which voids 
a boundary $\Dimless$-simplex of $\wdtild{K}$ is incident on. 
The task becomes complicated because a void may have disconnected boundaries 
and a $\Dimless$-simplex may bound more than one void. 
This is exemplified in Figure~\ref{fig:void_bd_recon1}.
To address this issue,
we orient the boundary $\Dimless$-simplices and determine the orientations
consistently from the voids they bound. 
This is possible
because an orientation of a $\Dimless$-simplex in $\real^\Dimtop$
associates exactly one of its two sides to the $\Dimless$-simplex.
To reconstruct the boundaries,
we first inspect the neighborhood of each $(\Dimlsls)$-simplex 
being a face of a boundary $\Dimless$-simplex
and pair the oriented boundary $\Dimless$-simplices in the neighborhood 
which locally bound the same void.
Figure~\ref{fig:edge-pair} gives an example of the oriented boundary $\Dimless$-simplices pairing
for $\Dimless=1$.
In Figure~\ref{fig:edge-pair},
there are three local voids each colored differently.
The oriented 1-simplices with the same color bound the same void and are paired.

After pairing the oriented boundary $\Dimless$-simplices, 
we group them by putting paired ones into the same group.
Each group then forms a $\Dimless$-cycle (with $\Zbb$ coefficients).
This is exemplified by Figure~\ref{fig:void_bd_recon} for $\Dimless=1$.
Note that in general,
the above grouping does not fully reconstruct the void boundaries.
This can be seen from Figure~\ref{fig:void_bd_recon1}
where the complex has four voids
but the grouping produces six 1-cycles.
In order to fully reconstruct the boundaries,
one has to retrieve the nesting structure of these $\Dimless$-cycles,
which may take $\OG(n^2)$ time in the worst-case. 
However, as we work on
a complex $\wdtild{K}$ that is $\Dimless$-connected, 
we cannot have voids with disconnected boundaries. 
Therefore, the 
grouping of oriented $\Dimless$-simplices
can fully recover the void boundaries.
Figure~\ref{fig:void_bd_recon2} gives an example for this when $\Dimless=1$,
where we add two 1-simplices to make the complex 1-connected.
The four 1-cycles produced by the grouping 
are exactly the boundaries of the four voids.

In the rest of this subsection,
we formalize the above ideas for reconstructing void boundaries
and provide a proof for the correctness.
Throughout this subsection, 
$\wdtild{K}$ and $\Dimless$ are as defined in Algorithm~\ref{algr:min-cyc-inf-bar}.
We first introduce the definition of the natural orientation of a $\diml$-simplex in $\real^\diml$.
We use its induced orientation to canonically orient the boundary simplices.

\begin{Definition}[Natural orientation~\cite{lee2010introduction}]
Let $\diml>1$ and ${\sG}=\Set{v_0,\ldots,v_\diml}$ 
be a $\diml$-simplex in $\real^\diml$,
an oriented simplex $\vec{\sG}=[v'_0,\ldots,v_\diml']$ of ${\sG}$ 
is {\it naturally oriented} if $\det(v'_1-v'_0,\ldots,v'_\diml-v'_0)>0$.
For each face $\sG'$ of $\sG$, 
the natural orientation of $\sG$ induces an orientation of $\sG'$ 
which we term as the {\it induced orientation}.
\end{Definition}

We now formally define the boundary of a void as follows:

\begin{Definition}[Boundary of void]\label{dfn:void}
Let $K$ be a simplicial complex embedded in $\real^\diml$ where $\diml\geq 2$,
an oriented $(\diml-1)$-simplex $\vec{\sG}^{\diml-1}=[v_0,\ldots,v_{\diml-1}]$ 
of ${K}$ is said to {\it bound} a void $\Vcal$ of $\real^\diml\smallsetminus|{K}|$ 
if the following conditions are satisfied:
\begin{itemize}
    \item The simplex ${\sG}^{\diml-1}=\Set{v_0,\ldots,v_{\diml-1}}$ is contained in the closure of $\Vcal$.
    \item Let $u$ be an interior point of ${\sG}^{\diml-1}=\Set{v_0,\ldots,v_{\diml-1}}$, 
    $v$ be a point in $\Vcal$ such that the line segment $\bar{uv}$ is contained in $\Vcal$ 
    and $\bar{uv}$ is orthogonal to the hyperplane spanned by ${\sG}^{\diml-1}$.
    Furthermore, let $\vec{\sG}^{\diml}$ be the naturally oriented simplex of $\Set{v,v_0,\ldots,v_{\diml-1}}$. 
    Then, $\vec{\sG}^{\diml-1}$ has the induced orientation from $\vec{\sG}^{\diml}$.
\end{itemize}
The {\it boundary} of a void $\Vcal$ is then defined as 
the set of oriented $(\diml-1)$-simplices of ${K}$ bounding $\Vcal$.
\end{Definition}

\begin{Remark}
We can also interpret the boundary of a void as a sum of oriented $(\diml-1)$-simplices, 
then the boundary defines a $(\diml-1)$-cycle (with $\Zbb$ coefficients).
\end{Remark}

We now describe the pairing algorithm of the oriented boundary $\Dimless$-simplices for $\wdtild{K}$.
From now on, we denote the set of boundary $\Dimless$-simplices of $\wdtild{K}$ as $\bd(\wdtild{K})$.
Let $\sG^{\Dimlsls}$ be a $(\Dimlsls)$-simplex 
which is a face of a $\Dimless$-simplex in $\bd(\wdtild{K})$, 
we first take a 2D plane $\DG$ which contains an interior point of $\sG^{\Dimless-1}$ 
and is orthogonal to the hyperplane spanned by $\sG^{\Dimlsls}$.
We then take the intersection of the plane $\DG$ with each boundary $\Dimless$-simplex 
in the neighborhood of $\sG^\Dimlsls$ to get a set of line segments that we order circularly
starting from an arbitrary one.
For each two consecutive line segments in this order which enclose a void,
we pick a point $p$ on the plane $\DG$ which resides in the void.
Suppose that one of the two line segments is derived from a boundary $\Dimless$-simplex
$\sG^\Dimless_0=\Set{v_0,\ldots,v_\Dimless}$.
We take the $(\Dimtop)$-simplex $\sG^\Dimtop=\Set{p,v_0,\ldots,v_\Dimless}$
and the induced oriented simplex $\vec{\sG}^\Dimless_0$ of $\sG^\Dimless_0$ derived
from the naturally oriented simplex of $\sG^\Dimtop$.
For the other line segment, we similarly derive an induced oriented simplex $\vec{\sG}^\Dimless_1$
and pair the two oriented $\Dimless$-simplices $\vec{\sG}^\Dimless_0$ and $\vec{\sG}^\Dimless_1$.
Figure~\ref{fig:edge-pair} can be reused to exemplify the pairing.
The union of the shaded regions in the figure is the plane $\DG$ and
$a$, $b$, $c$, and $d$ are the line segments derived from 
intersecting the plane with four boundary $\Dimless$-simplices.
Taking the circular order $a,b,c,d$, we see that
the consecutive ones which enclose a void are $(a,b)$, $(c,d)$, and $(d,a)$.
For $(a,b)$,
we can pick $p$ as an interior point in the blue region 
and the two oriented $\Dimless$-simplices corresponding to $a$ and $b$
can be induced and paired.

In summary, the steps of the \textsc{VoidBoundary} subroutine are the following:
\begin{enumerate}
\item For each $(\Dimlsls)$-simplex $\sG^{\Dimlsls}$ being a face of a $\Dimless$-simplex in $\bd(\wdtild{K})$,
pair all oriented boundary $\Dimless$-simplices in the neighborhood.
\item After gathering all the pairing,
group the oriented boundary $\Dimless$-simplices
by putting all paired ones into a group.
\item Return  $(\vec{\zG}_1,\ldots,\vec{\zG}_k)$,
each of which is a group of the oriented boundary $\Dimless$-simplices.
\end{enumerate}


The following theorem concludes the correctness of the reconstruction:

\begin{theorem}
\label{thm:walk-correct}
Any $\vec{\zG}_j$ returned by \textsc{VoidBoundary}
is the boundary of a void of $\real^\Dimtop\smallsetminus|\wdtild{K}|$.
\end{theorem}

\begin{proof}
See Appendix~\ref{apx:walk-correct-pf}.
\end{proof}

\subsection{Algorithm correctness}
\label{sec:corr-inf}

To prove the correctness of Algorithm~\ref{algr:min-cyc-inf-bar}, we need two conclusions about cycles with $\Zbb_2$ coefficients. 
Specifically, Proposition~\ref{lem:z2-cycle-simplex-bound-diff} says that an embedded $(q-1)$-cycle in $\mathbb{R}^q$ separates the space 
and hence the two oriented simplices of a $(q-1)$-simplex in the cycle 
bound different voids. 
Proposition~\ref{prop:exist-subcyc-d-conn} says that a $q$-simplex in a $q$-cycle belongs to a $q$-connected sub-cycle of the $q$-cycle. 

\begin{proposition} \label{lem:z2-cycle-simplex-bound-diff}
Let $\diml\geq 2$, 
$\zG$ be a $(\diml-1)$-cycle (with $\Zbb_2$ coefficients) of a simplicial complex embedded in $\real^q$,
and $\Zcal$ be the closure of the simplicial set $\zG$.
Then for any $(\diml-1)$-simplex $\sG$ of $\zG$, 
the two oriented simplices of $\sG$ must bound different voids of $\real^\diml\smallsetminus|\Zcal|$.
\end{proposition}
\begin{proof}
Consider a closed topological $\diml$-ball $\bll$ such that
$\sG\subseteq\bll$
and $\bll\intersect|\Zcal\smallsetminus{\sG}|$ equals the boundary of $\sG$.
Let $\bll_1$ and $\bll_2$ be the two open half balls of $\bll$ separated by $\sG$.
Then it is true that 
the two oriented simplices of $\sG$ bound different voids of $\real^\diml\smallsetminus|\Zcal|$
if and only if $\bll_1$ and $\bll_2$ are not connected in $\real^\diml\smallsetminus|\Zcal|$.
So we only need to show that $\bll_1$ and $\bll_2$ are not connected in $\real^\diml\smallsetminus|\Zcal|$.
Consider a filtration of $\Zcal$ where $\sG$ is the last simplex added.
Because $\sG$ is a positive simplex in the filtration,
by adding $\sG$,
the dimension of $\Hm_{\diml-1}$ must increase by 1.
By Alexander duality, 
the dimension of $\Hm_0$ of the complement space also increases by 1.
Then $\bll_1$ and $\bll_2$ cannot be connected in $\real^\diml\smallsetminus|\Zcal|$.
\end{proof}

\begin{proposition}
\label{prop:exist-subcyc-d-conn}
Let $\zG$ be a $\diml$-cycle (with $\Zbb_2$ coefficients) of a simplicial complex where $\diml>0$,
then for any $\diml$-simplex $\sG$ of $\zG$, 
there must be a $\diml$-cycle $\zeta'$ (with $\Zbb_2$ coefficients) containing $\sG$ 
such that $\zeta'\subseteq\zeta$ and $\zG'$ is $\diml$-connected.
\end{proposition}
\begin{proof}
We can construct an undirected graph $L$ for $\zeta$, with vertices of $L$ corresponding to the $\diml$-simplices in $\zeta$. 
For each $(\diml-1)$-simplex $\sG^{\diml-1}$ which is a face of a $\diml$-simplex of $\zeta$, 
let $\Ncal$ be the set of $\diml$-simplices in $\zeta$ having $\sG^{\diml-1}$ as a face, 
then $|\Ncal|$ must be even.
We can pair $\diml$-simplices of $\Ncal$ arbitrarily, 
and make each pair of $\diml$-simplices form an edge in $L$.
Let $C$ be the connected component of $L$ containing the corresponding vertex of $\sG$ 
and $\zeta'$ be the $\diml$-chain corresponding to $C$, then $\zeta'$ must be a cycle.
This is because 
we can pair the $(\diml-1)$-faces of all $\diml$-simplices in $\zeta'$ according to the edges in $L$,
so $\partial(\zeta')=0$.
Furthermore, $\zeta'$ contains $\sG$, $\zeta'\subseteq\zeta$, and $\zeta'$ is $\diml$-connected.
\end{proof}

Throughout the rest of this subsection,
some of the symbols we use refer to Algorithm~\ref{algr:min-cyc-inf-bar}.
We endow the ambient space $\real^\Dimtop$ with a ``cellular complex'' structure 
by treating voids of $\real^\Dimtop\smallsetminus|\wdtild{K}|$ as $(\Dimtop)$-dimensional ``cells''. 
This cellular complex of $\real^\Dimtop$ is denoted as $\Rcal^\Dimtop$ 
and $\Rcal^\Dimtop=\wdtild{K}\union \Set{\text{voids of }\real^\Dimtop\smallsetminus|\wdtild{K}|}$.
For $\Rcal^\Dimtop$, most terminologies from algebraic topology for simplicial complexes are inherited
with the exception that $(\Dimtop)$-dimensional elements of $\Rcal^\Dimtop$ are called {\it $(\Dimtop)$-cells}.
Then, we can also let $\thG$ denote the bijection 
from $(\Dimtop)$-cells of $\Rcal^\Dimtop$ to $V(G)$.
To derive $\partial(\Vcal)$ for a void $\Vcal$ of $\real^\Dimtop\smallsetminus|\wdtild{K}|$,
we map oriented $\Dimless$-simplices in the boundary of $\Vcal$  (Definition~\ref{dfn:void})
to their corresponding unoriented $\Dimless$-simplices.
Then $\partial(\Vcal)$ is defined as the sum (with $\Zbb_2$ coefficients) of these unoriented $\Dimless$-simplices.
It is not hard to see that $\partial(\Vcal)$ is a $\Dimless$-cycle (with $\Zbb_2$ coefficients)
because each void boundary is a $\Dimless$-cycle (with $\Zbb$ coefficients).

\begin{proposition}
\label{prop:algr-inf-bar-cut-pers}
For any cut $(S,T)$ of $(G,s_1,s_2)$, 
the $\Dimless$-chain $\zG=\thG\inv(\xi({S},{T}))$ 
is a persistent $\Dimless$-cycle of $[\birth,+\infty)$ and $w(\zG)=c(S,T)$. 
\end{proposition}
\begin{proof}
We have three things to show:
(i) $\zG$ contains $\sG_\birth^\Fcal$; 
(ii) $w(\zG)=c(S,T)$;
(iii) $\zG$ is a cycle.
Claim~(i) and~(ii) are not hard to verify and
we prove claim~(iii) by showing that $\zG=\sum_{\aG\in\thG\inv(S)}\partial(\aG)$, 
so that as a sum of cycles, $\zG$ is a cycle.
The detail for the equality of the two chains 
is omitted as it is similar to the one in the proof of Proposition~\ref{prop:algr-fin-bar-fincut-pers}.
\end{proof}

\begin{proposition}
\label{prop:algr-inf-bar-pers-cut}
For any persistent $\Dimless$-cycle $\zG$ of $[\birth,+\infty)$, 
there exists a cut $(S,T)$ of $(G,s_1,s_2)$ such that $c(S,T)\leq w(\zG)$.
\end{proposition}
\begin{proof}
Because of the nature of the pruning,
$\zG$ must reside in $K_\birth'$.
By Proposition~\ref{prop:exist-subcyc-d-conn}, 
there must be a $\Dimless$-cycle $\zG'\subseteq\zG$
such that $\zG'$ is $\Dimless$-connected and contains $\sG_\birth^\Fcal$.
Hence, $\zG'$ resides in $\wdtild{K}$.
Let $\Zcal'$ be the closure of the simplicial set $\zG'$,
we can run the void boundary reconstruction algorithm 
of Section~\ref{sec:void-bound-recon} on $\Zcal'$
and take a void boundary $\vec{\zG}$ containing an oriented simplex $\vec{\sG}_\birth^\Fcal$ of $\sG_\birth^\Fcal$.
We can map each oriented simplex of $\vec{\zG}$ to its unoriented simplex
and let $\zG_0$ be the sum of these unoriented simplices,
then $\zG_0$ is a $\Dimless$-cycle (with $\Zbb_2$ coefficients) and $\zG_0\subseteq\zG'$.
By Proposition~\ref{lem:z2-cycle-simplex-bound-diff}, 
the oppositely oriented simplex of $\vec{\sG}_\birth^\Fcal$ must not be in $\vec{\zG}$,
so $\zG_0$ contains $\sG_\birth^\Fcal$.
Let $\vec{\zG}$ bound a void $\Vcal$ of $\real^\Dimtop\smallsetminus|\Zcal'|$,
we can let $\Acal$ be the $(\Dimtop)$-chain of $\Rcal^\Dimtop$ consisting of all the $(\Dimtop)$-cells residing in $\Vcal$
and let $\Bcal$ be the $(\Dimtop)$-chain consisting of all the other $(\Dimtop)$-cells,
then $\partial(\Acal)=\partial(\Bcal)=\zG_0$.
Let $v_1,v_2$ be the two end vertices of $\thG(\sG_\birth^\Fcal)$.
Because the oppositely oriented simplex of $\vec{\sG}_\birth^\Fcal$ does not bound $\Vcal$ in $\Zcal'$,
it must be true that one of $v_1,v_2$ is in $\thG(\Acal)$ and the other is in $\thG(\Bcal)$.
We can let $(S,T)=(\thG(\Acal),\thG(\Bcal))$ or $(\thG(\Bcal),\thG(\Acal))$ 
based on which set contains the source of the flow network,
then $(S,T)$ is a cut of the flow network constructed in Algorithm~\ref{algr:min-cyc-inf-bar}.
Furthermore, we have $\zG_0=\thG\inv(\xi(S,T))$ and $c(S,T)=w(\zG_0)\leq w(\zG)$.
\end{proof}

The following theorem concludes the correctness of Algorithm~\ref{algr:min-cyc-inf-bar}: 

\begin{theorem}
Algorithm~\ref{algr:min-cyc-inf-bar} computes a minimal persistent $\Dimless$-cycle for the given interval $[\birth,+\infty)$.
\end{theorem}
\begin{proof}
First, the flow network $(G,s_1,s_2)$ constructed by Algorithm~\ref{algr:min-cyc-inf-bar} is valid.
The reason is that, by Proposition~\ref{lem:z2-cycle-simplex-bound-diff},
it cannot happen that the two oriented simplices of $\sG_\birth^\Fcal$ 
bound the same void of $\real^\Dimtop\smallsetminus|\wdtild{K}|$.
So $\sG_\birth^\Fcal$ must correspond to an edge of $G$.
Then by Proposition~\ref{prop:algr-inf-bar-cut-pers} and~\ref{prop:algr-inf-bar-pers-cut},
we can reach the conclusion.
\end{proof}


\section{Hardness for general complexes} \label{sec:hardness}

Similar to the work~\cite{chen2011hardness}, 
the NP-hardness proofs in this section accomplish the reduction with the help of a suspension operator. While Hatcher~\cite{hatcher2002algebraic} defines this operator for general topological spaces, we need a definition of the operator for simplicial complexes and observe some of its
properties that are useful for the proofs.

\subsection{Suspension operator}
\begin{Definition}[Suspension~\cite{ferrario2010simplicial}]
The {\it suspension} $\Sspn K$ of a simplicial complex $K$ 
is defined as a simplicial complex
\[
\begin{array}{l}
\Sspn K=\big\{\Set{\oG_1},\Set{\oG_2}\big\}\cup K\cup
\Big(\bigcup_{\sG\in K}\big\{\sG\cup\Set{\oG_1},\sG\cup\Set{\oG_2}\big\}\Big)
\end{array}
\]
where $\oG_1$, $\oG_2$ are two extra vertices.
\end{Definition}

\begin{Remark}
In the above definition, we denote a simplex by its set of vertices.
\end{Remark}


In the rest of this subsection,
we let $K$ be an arbitrary simplicial complex.
Any simplex of the form $\sG\cup\Set{\oG_i}$ in $\Sspn K$ is called a {\it suspended simplex}.
The symbol $\Sspn$ is also used to denote a linear map $\Sspn:\Chn_q(K)\to\Chn_{q+1}(\Sspn K)$, 
where $\Sspn \sG=\sG\cup\Set{\oG_1}+\sG\cup\Set{\oG_2}$ for any $q$-simplex $\sG$ of $K$. 
Note that since $\Sspn$ is injective, the map $\Sspn$ defines an isomorphism from $\Chn_q(K)$ to 
the image $\Sspn(\Chn_q(K))$. 
For any chain $A\in \Sspn(\Chn _q(K))$, we abuse the notation slightly by 
letting $\Sspn\inv A$ denote the chain in $\Chn_q(K)$ mapped to $A$ under $\Sspn$.

\begin{proposition}
\label{prop:isomorph-sus}
For any $q\geq 1$, the following diagram commutes:

\centerline{\xymatrix{
\Chn_q(K) \ar[d]_{\Sspn}^\approx \ar[r]^(.47){\partial} & \Chn_{q-1}(K) \ar[d]^{\Sspn}_\approx\\ 
\Sspn(\Chn_q(K)) \ar[r]^(.47){\partial} & \Sspn(\Chn_{q-1}(K)) \\ 
}}
\end{proposition}
\begin{proof}
For any $q$-simplex $\sG=\Set{v_0,\ldots,v_q}$ of $K$, we have
\begin{align*}
\partial(\Sspn\sG)
& =\partial\big(\Set{v_0,\ldots,v_q,\oG_1}+\Set{v_0,\ldots,v_q,\oG_2}\big) 
\\
& =\sum_{i=0}^q\Set{v_0,\ldots,\wdhat{v_i},\ldots,v_q,\oG_1}+\Set{v_0,\ldots,v_q}
+\sum_{i=0}^q\Set{v_0,\ldots,\wdhat{v_i},\ldots,v_q,\oG_2}+\Set{v_0,\ldots,v_q} 
\\
& =\sum_{i=0}^q\big(\Set{v_0,\ldots,\wdhat{v_i},\ldots,v_q,\oG_1}
+\Set{v_0,\ldots,\wdhat{v_i},\ldots,v_q,\oG_2}\big) 
\\
& =\sum_{i=0}^q\Sspn\big(\Set{v_0,\ldots,\wdhat{v_i},\ldots,v_q}\big)
=\Sspn\Bigg(\sum_{i=0}^q\Set{v_0,\ldots,\wdhat{v_i},\ldots,v_q}\Bigg)
  =\Sspn\partial(\sG)
\end{align*}
In the above equations, the notation $\wdhat{v_i}$ means that $v_i$ is deleted from the simplex.
\end{proof}

\begin{proposition}
\label{prop:sus-cyc-img}
For $q\geq 1$ and any $q$-cycle $\zG$ of $\Sspn K$ containing only suspended simplices, 
one has $\zG\in \Sspn(\Chn_{q-1}(K))$.
\end{proposition}
\begin{proof}
For any suspended $q$-simplex $\sG\union\Set{\oG_i}$ of $\zG$, 
if $\oG_i=\oG_1$, 
then $\sG\union\Set{\oG_2}$ must also belong to $\zG$ 
because no other suspended $q$-simplices of $\Sspn K$ have $\sG$ in the boundary. 
If $\oG_i=\oG_2$, the same argument follows.
\end{proof}

\begin{proposition}
\label{prop:sus-bound-img}
If $q$ is the top dimension of $K$ and $q\geq 1$, 
then for any $A\in \Chn_{q+1}(\Sspn K)$ 
such that $\partial(A)$ contains only suspended simplices,
one has $A\in\Sspn(\Chn_q(K))$.
\end{proposition}
\begin{proof}
Because $q$ is the top dimension of $K$, $A$ contains only suspended simplices.
For any $\sG\union\Set{\oG_i} \in A$, we have $\sG\in\partial\big(\sG\union\Set{\oG_i}\big)$.
If $\oG_i=\oG_1$, 
to make $\sG$ cancelled in $\partial(A)$,
$\sG\union\Set{\oG_2}$ must also belong to $A$ 
because no other $(q+1)$-simplices in $\Sspn K$ have $\sG$ in the boundary. 
If $\oG_i=\oG_2$, the same argument follows.
\end{proof}

\subsection{Hardness for finite intervals}
The following proposition helps to prove our conclusion of the hardness:

\begin{proposition}
\label{lem:nph-fin}
{PCYC-FIN}$_{\Dim-1}$ reduces to {PCYC-FIN}$_\Dim$ for $\Dim\geq 2$.
\end{proposition}
\begin{proof}
Given an instance $(K,\Fcal,[\birth,\death))$ of PCYC-FIN$_{\Dim-1}$,
where the $i_\text{th}$ complex of $\Fcal$ is denoted as $K_i$,
we can assume the top dimension of $K$ to be $\Dim$. 
The reason is that
if it were not, we can restrict $\Fcal$ to the $\Dim$-skeleton of $K$ without affecting
$\persdgm_{\Dim-1}(\Fcal)$ 
and the persistent $(\Dim-1)$-cycles. 
Then, we let $\Sspn K$ be the simplicial complex for the instance of PCYC-FIN$_\Dim$ we are going to construct.
For any suspended $\Dim$-simplex $\sG\cup\Set{\oG_i}$ of $\Sspn K$,
let the weight of $\sG\cup\Set{\oG_i}$ be
half of the weight of $\sG$ in $K$.
Furthermore,
let the weight of any non-suspended $\Dim$-simplex of $\Sspn K$ 
be the sum of all the weights of $(\Dim-1)$-simplices in $K$ plus $1$.
We endow $\Sspn K$ with a filtration
$\Sspn\Fcal: \emptyset=\wdhat{K}_0\subseteq \wdhat{K}_1\subseteq\ldots\subseteq \wdhat{K}_{3n+2}=\Sspn K$, 
where $n$ is the number of simplices of $K$. 
Denoting the ${i}_\text{th}$ simplex added in $\Fcal$ as $\sG_i$ 
and the ${i}_\text{th}$ simplex added in $\Sspn\Fcal$ 
as $\wdhat{\sG}_i$, 
we let  $\wdhat{\sG}_1=\Set{\oG_1}$, $\wdhat{\sG}_2=\Set{\oG_2}$, 
and for any $1\leq i\leq n$, 
$\wdhat{\sG}_{3i}=\sG_i$, $\wdhat{\sG}_{3i+1}=\sG_i\cup\Set{\oG_1}$,
$\wdhat{\sG}_{3i+2}=\sG_i\cup\Set{\oG_2}$. 


We observe the following facts: 
\begin{enumerate}[label=(\roman*)]
    \item For any $i$, $\wdhat{\sG}_{3i}$ is positive and pairs with $\wdhat{\sG}_{3i+1}$ in $\Sspn\Fcal$.
    \label{item:trivial-pair}
    \item For any $i$ and $j$, if there is a $(\Dim-1)$-cycle created by $\sG_i$ which is a boundary in $K_j$,
    then there is a $\Dim$-cycle created by $\wdhat{\sG}_{3i+2}$ which is a boundary in $\wdhat{K}_{3j+2}$.
    \label{item:norm-pcyc-to-sspn}
    \item For any $i$ and $j$, 
    if there is a $\Dim$-cycle created by $\wdhat{\sG}_{3i+2}$ which is a boundary in $\wdhat{K}_{3j+2}$,
    then there is a $(\Dim-1)$-cycle created by $\sG_i$ which is a boundary in $K_j$.
    \label{item:sspn-pcyc-to-norm}
\end{enumerate}
The correctness of (i) is not hard to verify.
To verify (ii),
we can suspend the $(\Dim-1)$-cycle and use Proposition~\ref{prop:isomorph-sus}
to reach the claim.
The argument for (iii) is as follows:
Consider a $\Dim$-cycle $\wdhat{\zG}_0$ created by $\wdhat{\sG}_{3i+2}$ which is a boundary in $\wdhat{K}_{3j+2}$.
For any non-suspended $\Dim$-simplex $\sG$ of $\wdhat{\zG}_0$, 
we add $\partial\big(\sG\cup\Set{\oG_1}\big)$ to the cycle $\wdhat{\zG}_0$ 
so that $\sG$ is canceled and only suspended simplices are added. 
Note that the adding process only adds $\Dim$-simplices in $\wdhat{K}_{3i+2}$ 
and never cancels $\wdhat{\sG}_{3i+2}$.
After all non-suspended simplices of $\wdhat{\zG}_0$ are canceled,
we derive a $\Dim$-cycle $\wdhat{\zG}$
which is created by $\wdhat{\sG}_{3i+2}$
and contains only suspended simplices.
By Proposition~\ref{prop:sus-cyc-img}, 
$\Sspn\inv\wdhat{\zG}$ is well defined.
Since $\wdhat{\zG}$ is homologous to $\wdhat{\zG}_0$ in $\wdhat{K}_{3i+2}$, 
$\wdhat{\zG}$ is also a boundary in $\wdhat{K}_{3j+2}$. 
Let $\wdhat{\zG}$ be the boundary of a $(\Dim+1)$-chain $\wdhat{A}$ in $\wdhat{K}_{3j+2}$.
Because $\Sspn K_j=\wdhat{K}_{3j+2}$,
by Proposition~\ref{prop:sus-bound-img},
$\wdhat{A}\in\Sspn(\Chn_{\Dim}(K_j))$.
Furthermore, by Proposition~\ref{prop:isomorph-sus},
we have $\Sspn\inv\wdhat{\zG}=\Sspn\inv\partial(\wdhat{A})=\partial(\Sspn\inv\wdhat{A})$.
So $\Sspn\inv\wdhat{\zG}$ is a $(\Dim-1)$-cycle created by $\sG_i$ which is a boundary in $K_j$. 

From the above facts, it is immediate that $\wdhat{\sG}_{3\birth+2}$ is a positive simplex in $\Sspn\Fcal$ and pairs with $\wdhat{\sG}_{3\delta+2}$ so that $[3\birth+2,3\death+2)$ 
is an interval in $\persdgm_{\Dim}(\Sspn\Fcal)$.
It is also true that there is a bijection from the persistent $(\Dim-1)$-cycles of $[\birth,\death)$
to the persistent $\Dim$-cycles of $[3\birth+2,3\death+2)$ containing only suspended simplices.
Furthermore, the bijection preserves the weights of the cycles.
From the weight assigning policy,
the minimal persistent $\Dim$-cycle of $[3\birth+2,3\death+2)$
must contain only suspended simplices,
so this minimal persistent $\Dim$-cycle of $[3\birth+2,3\death+2)$
induces a minimal persistent $(\Dim-1)$-cycle of $[\birth,\death)$.
Now we have reduced PCYC-FIN$_{\Dim-1}$ to PCYC-FIN$_\Dim$. 
Furthermore, the reduction is in polynomial time and the size of $(\Sspn K,\Sspn\Fcal,[3\birth+2,3\death+2))$ 
is a polynomial function of the size of $(K,\Fcal,[\birth,\death))$.
\end{proof}

We have the following result from~\cite{dey19pers1cyc}:
\begin{proposition}
\label{prop:nph-fin-dim1}
{PCYC-FIN}$_1$ is NP-hard.
\end{proposition}
Combining Proposition~\ref{lem:nph-fin} and~\ref{prop:nph-fin-dim1}, 
we obtain the following theorem:

\begin{theorem}
{PCYC-FIN}$_\Dim$ is NP-hard for $\Dim\geq 1$.
\end{theorem}

\subsection{Hardness for infinite intervals}\label{sec:inf-hard}

In this subsection, we prove that it is NP-hard to approximate WPCYC-INF$_\Dim$ with any fixed ratio.
Let PROB be a minimization problem with solutions having positive costs.
Given an instance $\Ical$ of PROB,
let $C^*$ be the cost of the minimal solution of $\Ical$.
For $r\geq 1$,
a solution of $\Ical$ with cost $C$ is said to have an {\it approximation ratio $r$}
if $C/C^*\leq r$~\cite{CLRS-approx}.
We let PROB$[r]$ denote the problem that asks for an approximate solution
with ratio $r$ given an instance of PROB.
Moreover, in order to make approximation ratios well-defined for WPCYC-INF$_\Dim$,
we let WPCYC-INF$_\Dim^+$
denote a subproblem of WPCYC-INF$_\Dim$ 
where all $\Dim$-simplices are positively weighted.

Before proving the hardness result, we first recall the definition of the nearest codeword problem, 
which is NP-hard to approximate with any fixed ratio~\cite{chen2011hardness}:

\begin{problem}[NR-CODE]
Given an $l\by k$ full-rank matrix $\Acal$ over $\Zbb_2$ for $k<l$ and a vector $y_0\in (\Zbb_2)^l \smallsetminus \img(\Acal)$, 
find a vector in $y_0+\img(\Acal)$ with the minimal Hamming weight.
\end{problem}

\begin{Remark}
The Hamming weight of a vector $y$, denoted as $\|y\|_H$,
is the number of non-zero components in $y$.
\end{Remark}

\begin{theorem}\label{thm:WPCYC-INF-2-hard}
{WPCYC-INF}$^{\, +}_2$ is NP-hard to approximate with any fixed ratio.
\end{theorem}

Similar to the NP-hardness proof of homology localization in~\cite{chen2011hardness}, 
our proof of Theorem~\ref{thm:WPCYC-INF-2-hard} conducts the reduction from the NR-CODE problem.
One may think that a direct reduction from homology localization may be more straightforward.
However, such a reduction is not immediately evident.
The two problems appear to be of different nature:
While the homology localization problem asks for a minimal cycle in a given homology class,
{WPCYC-INF}$^{+}_2$ asks for a minimal cycle in a complex containing a given simplex 
without referring to any particular homology class.

\begin{proof}
For any $r> 1$, we reduce the NP-hard problem NR-CODE$[2r]$ to WPCYC-INF$^+_2[r]$.
Given an instance $(\Acal,y_0)$ of NR-CODE$[2r]$, 
we first compute the $(l-k)\by l$ parity check matrix $\Acal^\perp$~\cite{chen2011hardness}, 
which is a matrix such that $\ker(\Acal^\perp)=\img(\Acal)$.
Similar to the proof of Lemma 4.3.1 in~\cite{chen2011hardness}, 
we then build a ``tube complex'' $T_1$ with $(l-k)$ 1-cells each of which is a 1-sphere 
and $l$ 2-cells each of which is a 2-sphere with holes. 
The 2-cells of $T_1$ are attached to the 1-cells along the holes such that the boundary matrix $\partial_2$ of this tube complex equals $\Acal^\perp$.
The ``$\diml$-chains'' and ``$\diml$-cycles'' for a tube complex 
are analogously defined as for a simplicial complex.
We also assign a weight of 1 to each 2-cell of $T_1$.
By this construction, there is a straightforward bijection $\phi:(\Zbb_2)^l\to\Chn_2(T_1)$, 
such that the Hamming weight of a vector equals the weight of the corresponding 2-chain.
Note that $\Zyc_2(T_1)=\ker(\partial_2)=\phi(\ker(\Acal^\perp))=\phi(\img(\Acal))$.
Let $\wdtild{y}_0=\phi(y_0)$,
we then add a 2-cell $\wdhat{t}$ whose boundary equals $\partial_2(\wdtild{y}_0)$
to $T_1$ and get a new tube complex $T_2$.
We call the 2-cycles in $T_2$ which are not in $T_1$ as the new 2-cycles in $T_2$.
Then $\wdhat{t}+\wdtild{y}_0$ is a new 2-cycle in $T_2$
and the set of new 2-cycles in $T_2$ is $\wdhat{t}+\wdtild{y}_0+\Zyc_2(T_1)$.
We let the weight of $\wdhat{t}$ also be 1.
Note that there is a bijection $\psi:y_0+\img(\Acal)\to\wdhat{t}+\wdtild{y}_0+\Zyc_2(T_1)$,
where $\psi(y_0+z)=\wdhat{t}+\wdtild{y}_0+\phi(z)$
for any $z\in \img(\Acal)$,
such that $w(\psi(y_0+z))=\|y_0+z\|_H+w(\wdhat{t})$.


We then construct an instance of WPCYC-INF$^+_2[r]$ 
by first triangulating $T_2$ to get a simplicial complex $K$.
We make $K$ 2-weighted such that 
the sum of the weights of all triangles in any 2-cell of $T_2$ equals 
the weight of the 2-cell.
It is not hard to make the size of $K$ a polynomial function of the number of cells of $T_2$.
Let $\sG$ be a 2-simplex in the triangulation of the 2-cell $\wdhat{t}$.
We build a filtration $\Fcal$ of $K$ with $\sG$ being the last simplex added. 
Let the index of $\sG$ in $\Fcal$ be $\birth$.
Then, $[\birth,+\infty)$ is an infinite interval of $\persdgm_2(\Fcal)$. 
Note that there is a bijection between the new 2-cycles in $T_2$ 
and the persistent 2-cycles of $[\birth,+\infty)$,
where the weights of the cycles are preserved.
Therefore, from the solution of WPCYC-INF$^+_2[r]$ with the input $(K,\Fcal,[\birth,+\infty))$,
we can derive a new 2-cycle $\wdhat{t}+\wdtild{y}_0+\zG$ of $T_2$,
where $\zG\in\Zyc_2(T_1)$
and $\wdhat{t}+\wdtild{y}_0+\zG$ is an $r$-approximation of the minimal new 2-cycle.
Let $\wdhat{t}+\wdtild{y}_0+\zG^*$ be a minimal new 2-cycle of $T_2$,
we have
\[\frac{w(\wdhat{t}+\wdtild{y}_0+\zG)}{w(\wdhat{t}+\wdtild{y}_0+\zG^*)}\leq r
\implies
\frac{w(\wdhat{t})+w(\wdtild{y}_0+\zG)}{w(\wdhat{t})+w(\wdtild{y}_0+\zG^*)}\leq r
\implies
w(\wdtild{y}_0+\zG)\leq r-1+rw(\wdtild{y}_0+\zG^*)\]
We also have
\[1\leq \frac{r}{r-1}w(\wdtild{y}_0+\zG^*)\implies r-1\leq rw(\wdtild{y}_0+\zG^*) \]
Therefore
\[w(\wdtild{y}_0+\zG)\leq 2rw(\wdtild{y}_0+\zG^*)\implies \|y_0+\phi\inv(\zG)\|_H\leq 2r\|y_0+\phi\inv(\zG^*)\|_H \]
Since $y_0+\phi\inv(\zG^*)$ is a minimal solution of $(\Acal,y_0)$,
then $y_0+\phi\inv(\zG)$ is a $2r$-approximation of the minimal solution of $(\Acal,y_0)$.
Hence, we have reduced NR-CODE$[2r]$ to WPCYC-INF$^+_2[r]$.
Furthermore, the reduction is in polynomial time and the sizes of the instances are related by a polynomial function, 
so WPCYC-INF$^+_2[r]$ is NP-hard.
\end{proof}

\begin{theorem}
{WPCYC-INF}$^{\, +}_\Dim$ is NP-hard to approximate with any fixed ratio for $\Dim\geq 2$.
\end{theorem}
\begin{proof}
For any $\Dim\geq 3$ and $r\geq 1$,
we reduce WPCYC-INF$^+_{\Dim-1}[r]$ to WPCYC-INF$^+_\Dim[r]$.
Given an instance $(K,\Fcal,[\birth,+\infty))$ of WPCYC-INF$^+_{\Dim-1}[r]$,
where the $i_\text{th}$ complex of $\Fcal$ is denoted as $K_i$,
let $K'=\Sspn K_\birth^{\Dim-1}$ where
$K_\birth^{\Dim-1}$ is the $(\Dim-1)$-skeleton of $K_\birth$.
We make $K'$ $\Dim$-weighted such that 
any $\Dim$-simplex $\sG\cup\Set{\oG_i}$ of $K'$ has half of the weight of $\sG$ in $K$.
The complex $K'$ is endowed with a filtration $\Fcal'$
such that $\sG_\birth^\Fcal\cup\Set{\oG_2}$ is the last simplex added to $\Fcal'$.
Let $\birth'$ be the index of $\sG_\birth^\Fcal\cup\Set{\oG_2}$ in $\Fcal'$, 
then $[\birth',+\infty)\in\persdgm_\Dim(\Fcal')$.
It is true that $\Sspn$ restricts to a bijection from $\Zyc_{\Dim-1}(K_\birth)$ to $\Zyc_\Dim(K')$
preserving the weights of the cycles.
Furthermore, for any $\zG\in\Zyc_{\Dim-1}(K_\birth)$,
$\zG$ is a persistent $(\Dim-1)$-cycle of $[\birth,+\infty)\in\persdgm_{\Dim-1}(\Fcal)$
if and only if $\Sspn\zG$ is a persistent $\Dim$-cycle of $[\birth',+\infty)\in\persdgm_\Dim(\Fcal')$.
Suppose that $\zG'$ is a solution for the instance $(K',\Fcal',[\birth',+\infty))$ of WPCYC-INF$^+_\Dim[r]$,
i.e., $\zG'$ is an $r$-approximation of the minimal solution.
Then, $\Sspn\inv\zG'$ is an $r$-approximation 
for the instance $(K,\Fcal,[\birth,+\infty))$ of WPCYC-INF$^+_{\Dim-1}[r]$.
Therefore, the reduction is done.
\end{proof}

\section{Experimental results}\label{sec:exp}

We experiment with our algorithms for WPCYC-FIN$_2$ and WEPCYC-INF$_2$ on several volume datasets. 
Since volume data have a natural cubical complex structure, we adapt 
our implementation slightly in order to work on cubical complexes.
The cubical complex for volume data consists of cells in dimensions from 0 to 3 
with the underlying space homeomorphic to a 3-dimensional ball.
Note that a filtration built from a volume dataset 
does not produce any infinite intervals.
Hence, in order to test our algorithm for WEPCYC-INF$_2$,
we take a finite interval and compute the minimal 2-cycle born at the birth time,
which is exactly what WEPCYC-INF$_2$ computes.
We use the \texttt{Gudhi}~\cite{gudhi:urm} library to build the filtrations 
and compute the persistence intervals.
From the experiments, 
we can see that the minimal persistent 2-cycles computed by our algorithms capture various features of
the data which originate from different fields.
Note that the combustion, hurricane, and medical datasets 
are time-varying 
and we chose a single time frame to compute the persistent intervals and cycles.

\begin{figure}
\centering
  \subfloat[]{\includegraphics[width=0.25\linewidth]{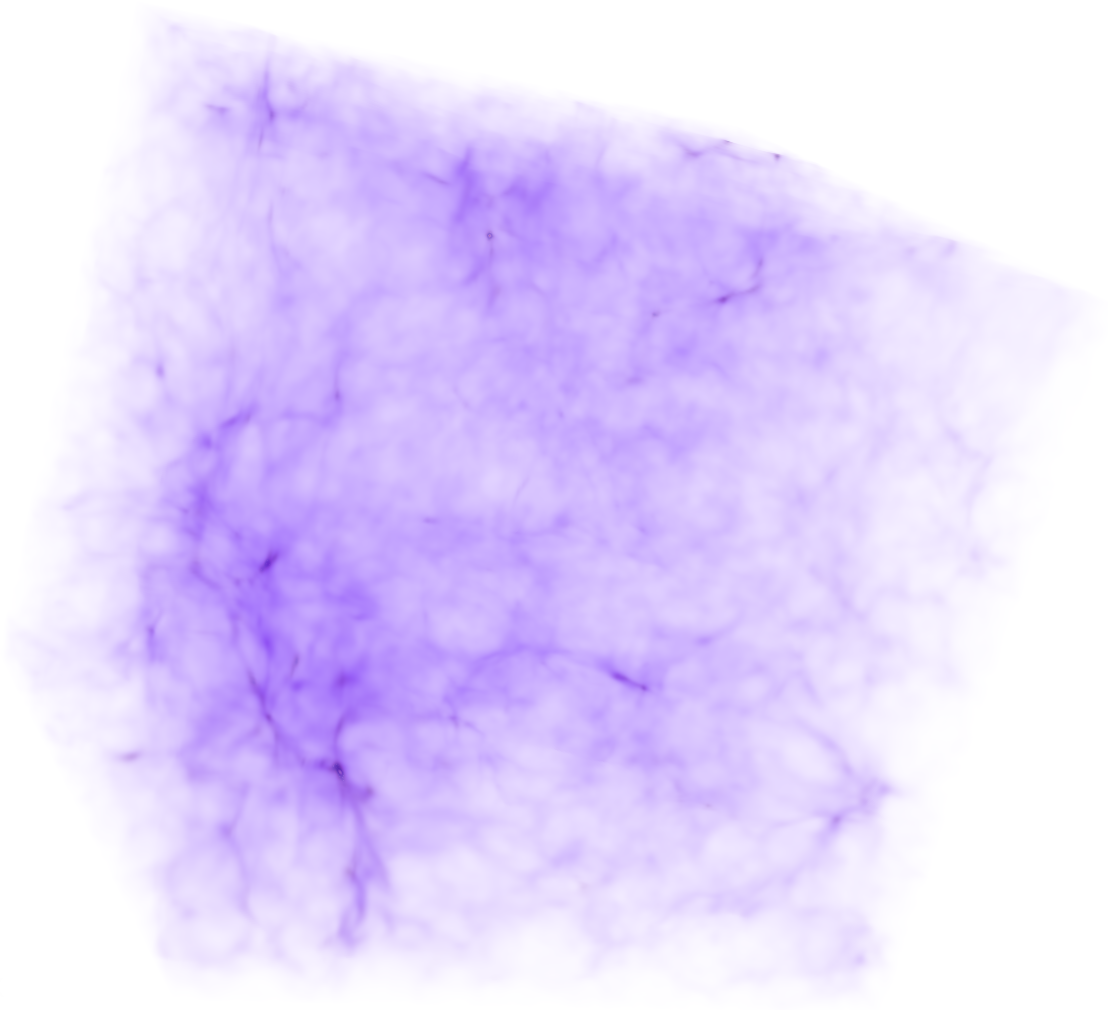}\label{fig:cos1}}
  \subfloat[]{\includegraphics[width=0.25\linewidth]{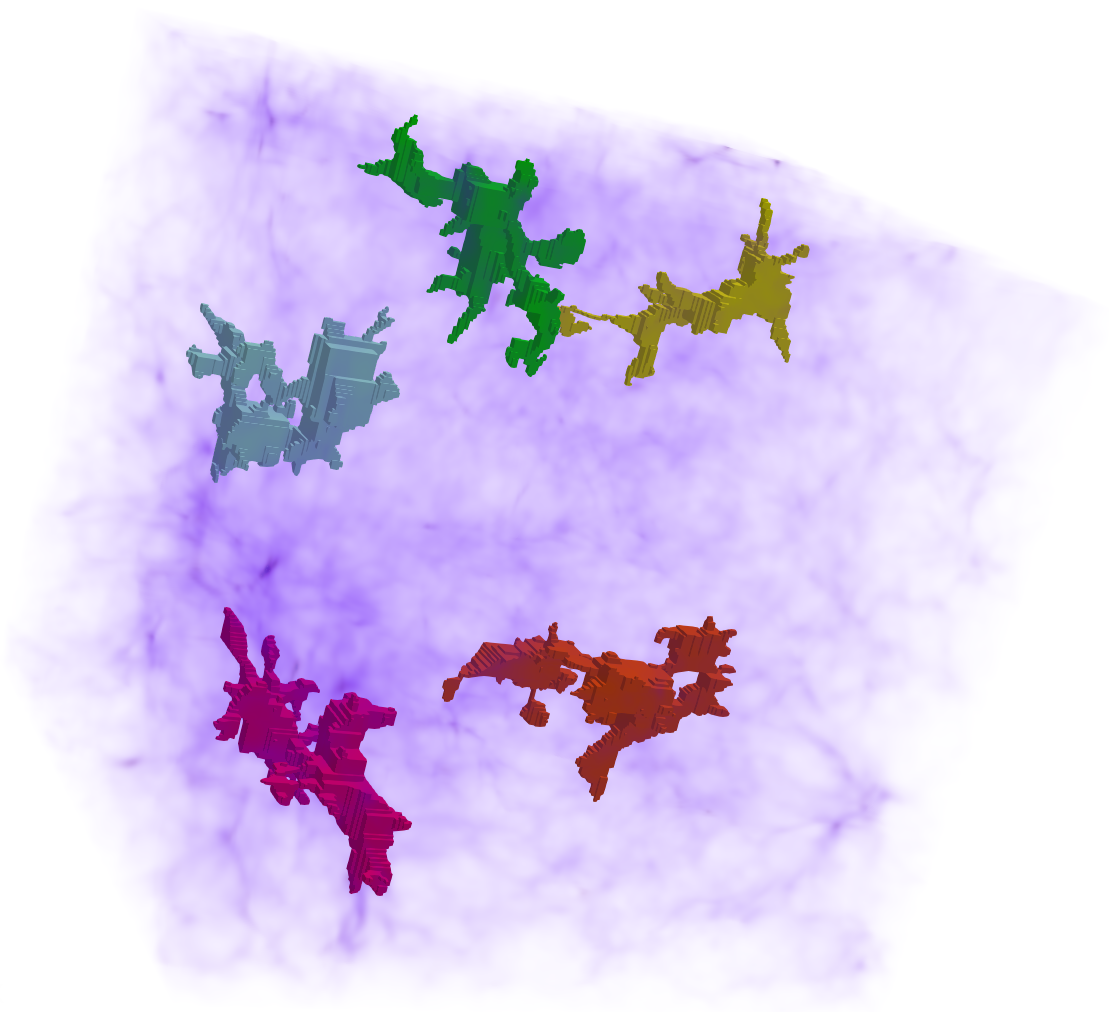}\label{fig:cos2}}
  \subfloat[]{\includegraphics[width=0.25\linewidth]{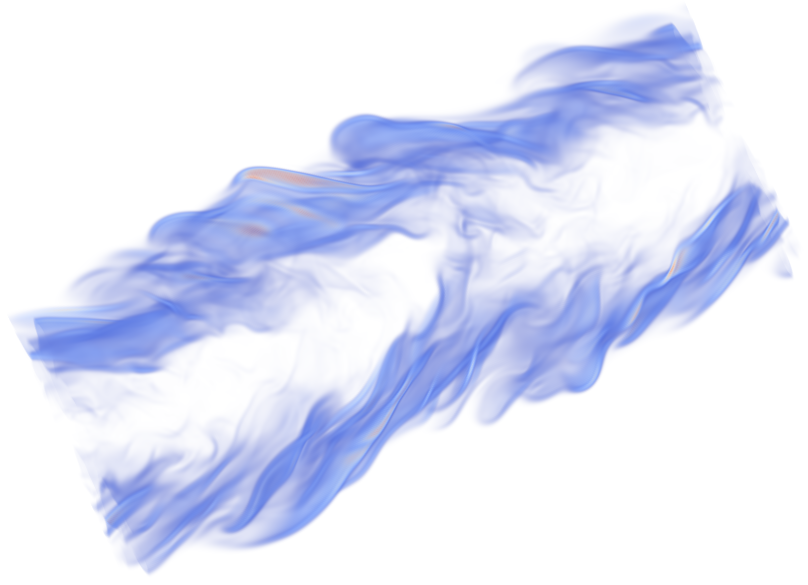}\label{fig:comb1}}
  \hspace{0.2em}
  \subfloat[]{\includegraphics[width=0.20\linewidth]{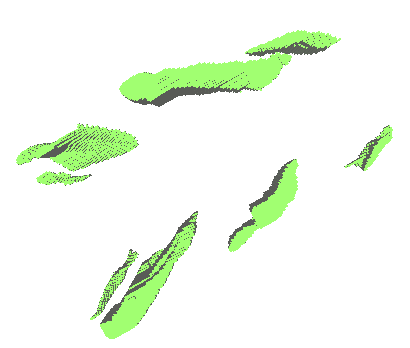}\label{fig:comb2}}
   \caption{
   (a,b) Cosmology dataset and the minimal persistent 2-cycles of the top five longest intervals.
   (c,d) Turbulent combustion dataset and its corresponding minimal persistent 2-cycles.
   }
    \label{fig:results1}
\end{figure}

\paragraph{Cosmology.} 
The simulation data shown in Figure~\ref{fig:cos1} from computational cosmology~\cite{Almgren_2013}
consist of dark matter represented as particles.
The thread-like structures in deep purple shown in Figure~\ref{fig:cos1} 
correspond to sites of large scale structure formation. 
Galaxy clusters/superclusters are contained in such large scale structures. 
Figure~\ref{fig:cos2} shows the minimal persistent 2-cycles of the top five longest intervals computed by our algorithms 
and these cycles precisely represent the top five galaxy clusters/superclusters in volume.

\paragraph{Combustion.} 
The data shown in Figure~\ref{fig:comb1}
correspond to the physical variable\footnote{
  A physical variable defines a scalar value of a certain kind on each point.} 
$\chi$ from a model of a turbulent combustion process.
The variable $\chi$ represents scalar dissipation rate and
provides a measure of the maximum possible chemical reaction rate. 
The minimal persistent 2-cycles 
shown in Figure~\ref{fig:comb2}
represent areas with high value of $\chi$.

\begin{figure}
    \centering
    \subfloat[]{\includegraphics[width=0.14\linewidth]{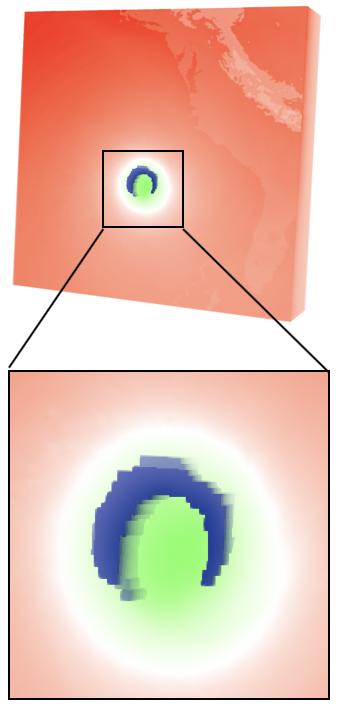}\label{fig:isa2}}
     \hspace{3em}
    \subfloat[]{\includegraphics[width=0.24\linewidth]{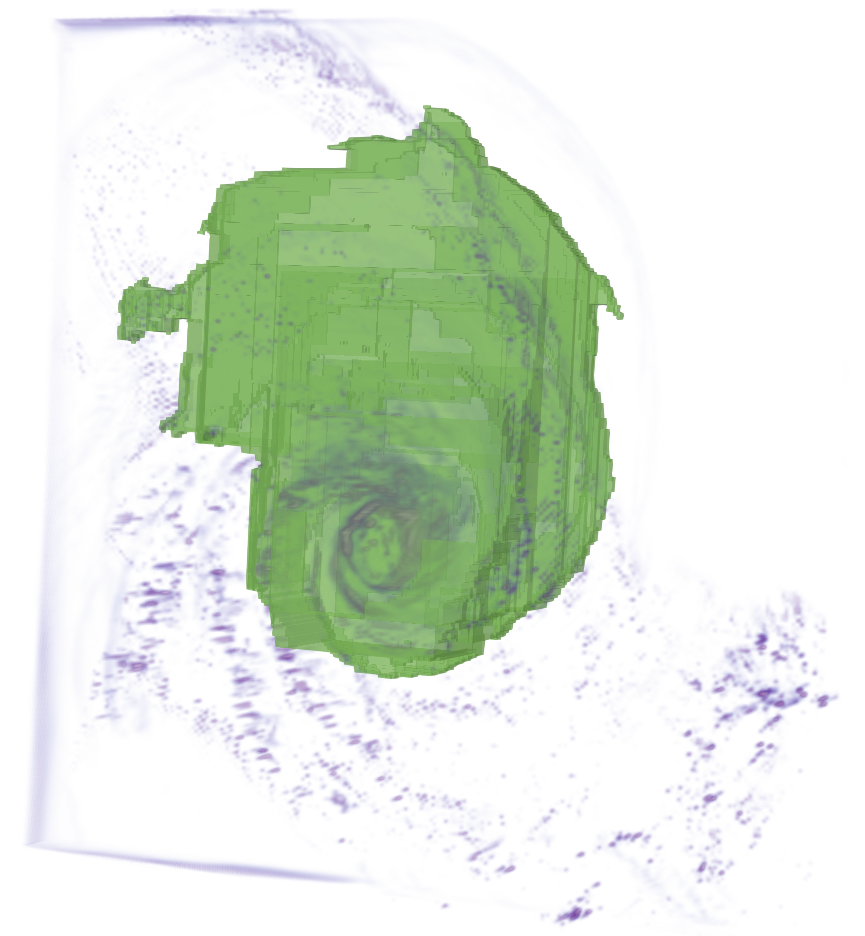}\label{fig:isa1}}
     \hspace{1em}
    \subfloat[]{\includegraphics[width=0.40\linewidth]{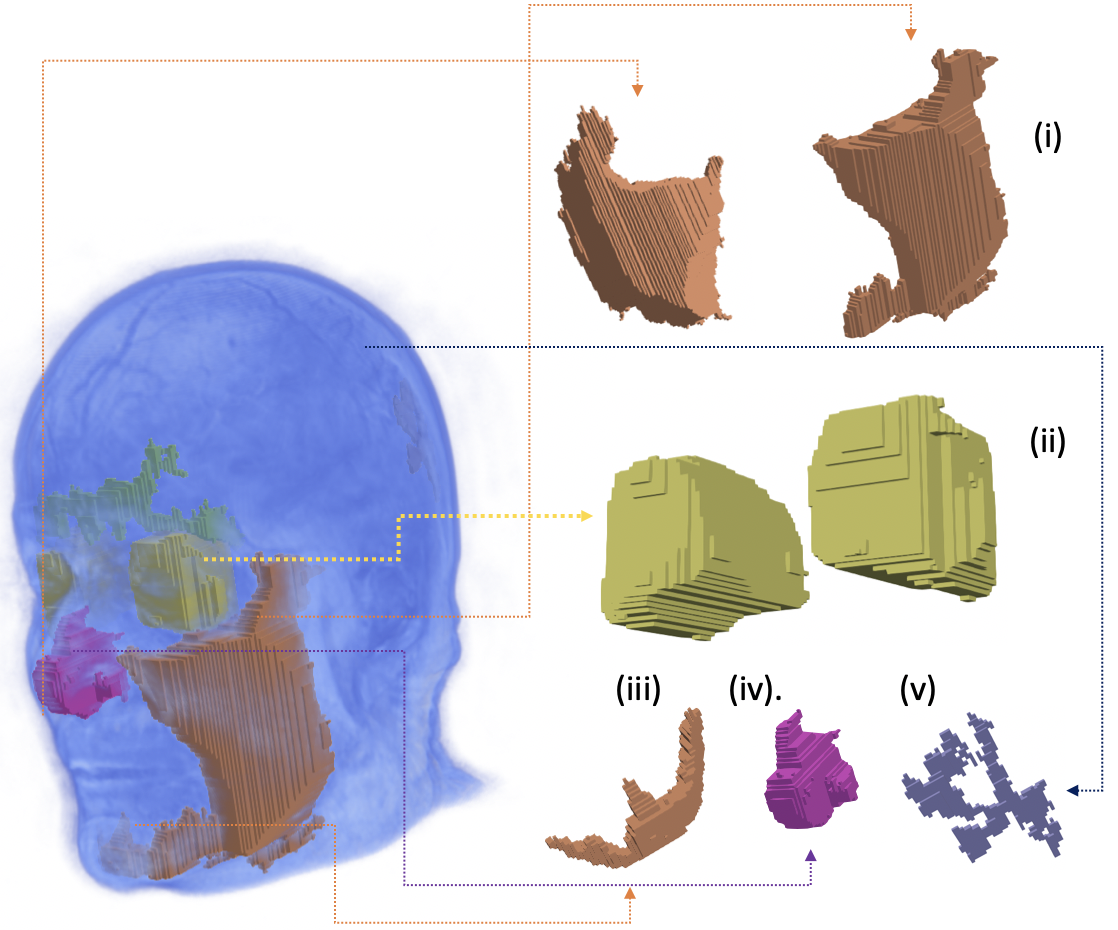}\label{fig:mri1}}
    \caption{
    (a,b) Minimal persistent 2-cycles for the hurricane model. 
    (c) Minimal persistent 2-cycles of the larger intervals for the human skull. 
    \textsf{i}: Right and left cheek muscles with the right one rotated for better visibility.
    \textsf{ii}: Right and left eyes.  
    \textsf{iii}: Jawbone. \textsf{iv}: Nose cartilage. 
    \textsf{v}: Nerves in the parietal lobe.
    }
    
    \label{fig:results2}
\end{figure}

\paragraph{Hurricane.} 
This dataset\footnote{
  The Hurricane Isabel data is produced by the Weather Research and Forecast (WRF) model, courtesy of NCAR, and the U.S. National Science Foundation (NSF).}
with $11$ physical variables
corresponds to the devastating hurricane named Isabel. 
We down-sampled the data into a resolution of $250\times 250 \times 50$ 
and worked with two physical variables.
The minimal persistent 2-cycle colored blue in Figure~\ref{fig:isa2} 
is computed on the cloud-volume variable 
and extracts the eye of the hurricane.
The minimal persistent 2-cycle colored green in Figure~\ref{fig:isa1}
is computed on the pressure variable
and captures the jagged shape of the pressure variation around the hurricane.

\paragraph{Medical imaging.} 
This dataset from the ADNI~\cite{petersen2010alzheimer} project contains the MRI scan of a healthy human skull. 
The minimal persistent 2-cycles corresponding to the larger intervals as shown in Figure~\ref{fig:mri1}
are computed from two time frames.
They extract significant features such as eyes, cartilages, nerves, and muscles.

\begin{figure}
    \centering
    \subfloat[]{\includegraphics[width=0.3\linewidth]{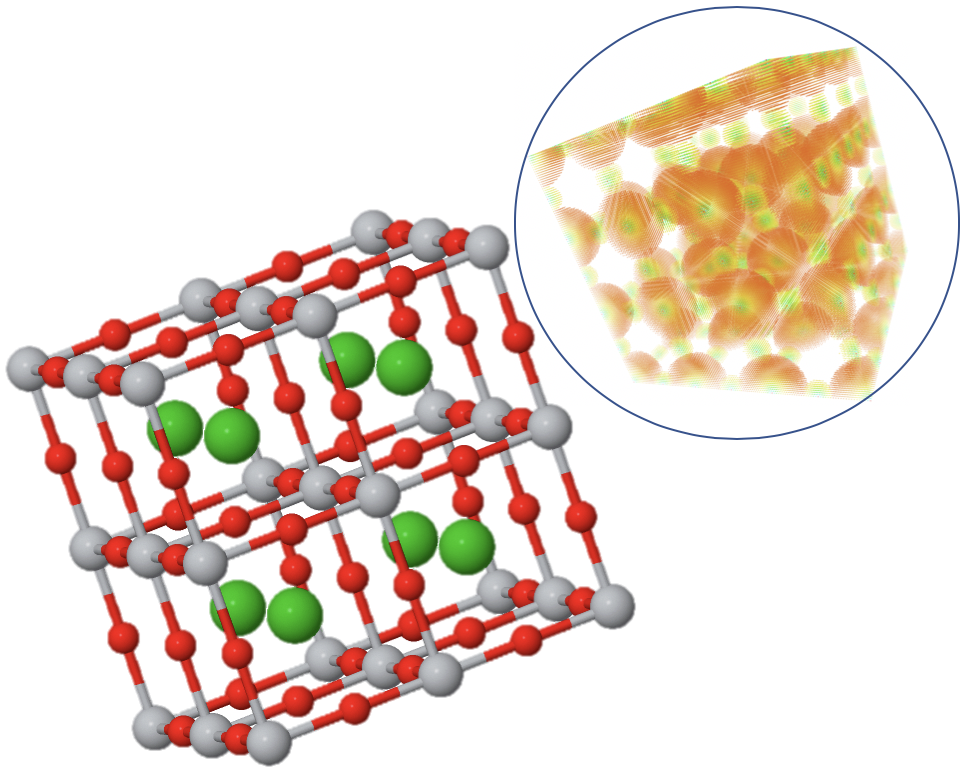}\label{fig:cubicStruct}}
      \hspace{0.5em}
   \subfloat[]{\includegraphics[width=0.20\linewidth]{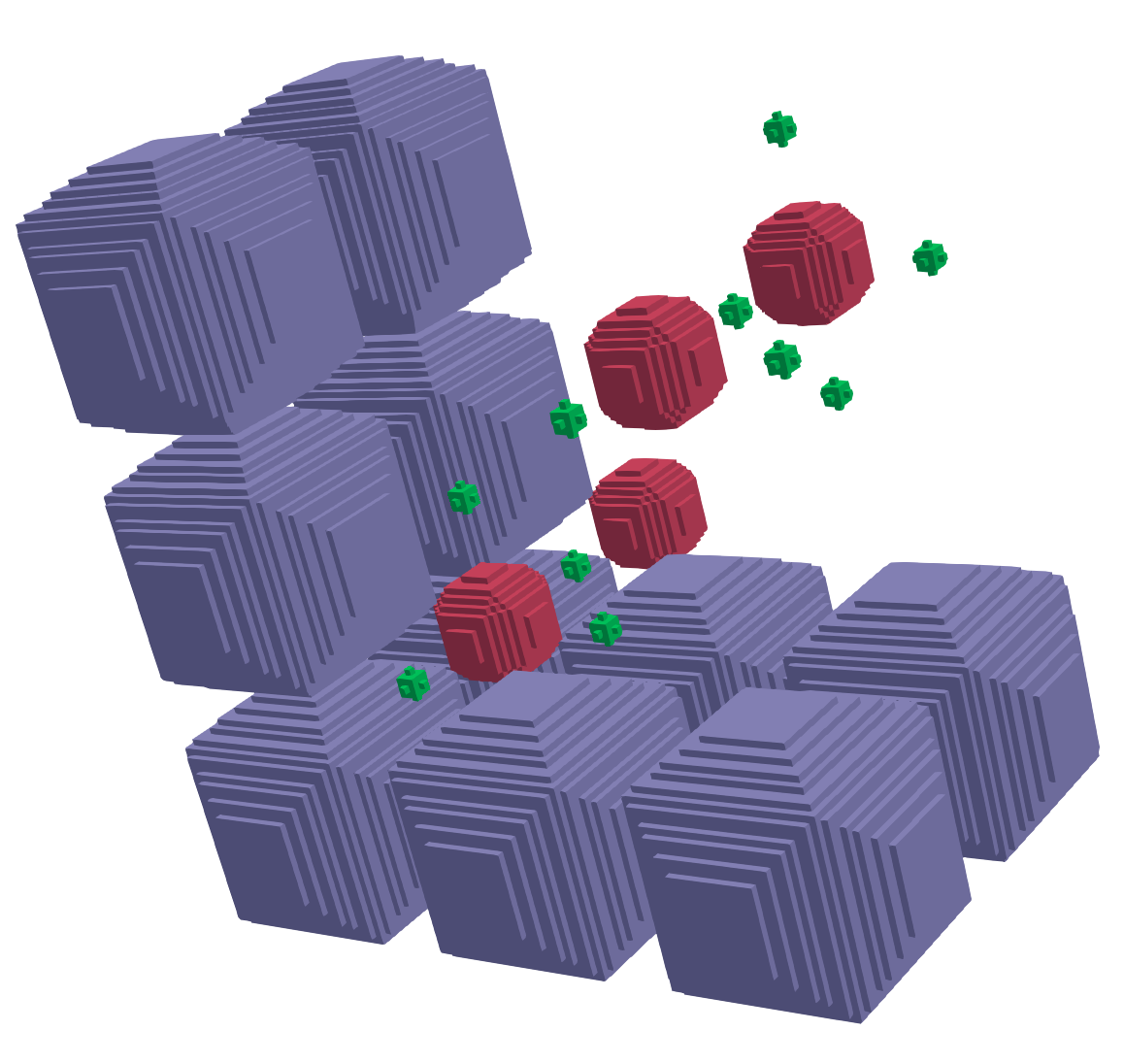}\label{fig:cubic0}}
   \hspace{0.5em}
   \subfloat[]{\includegraphics[width=0.20\linewidth]{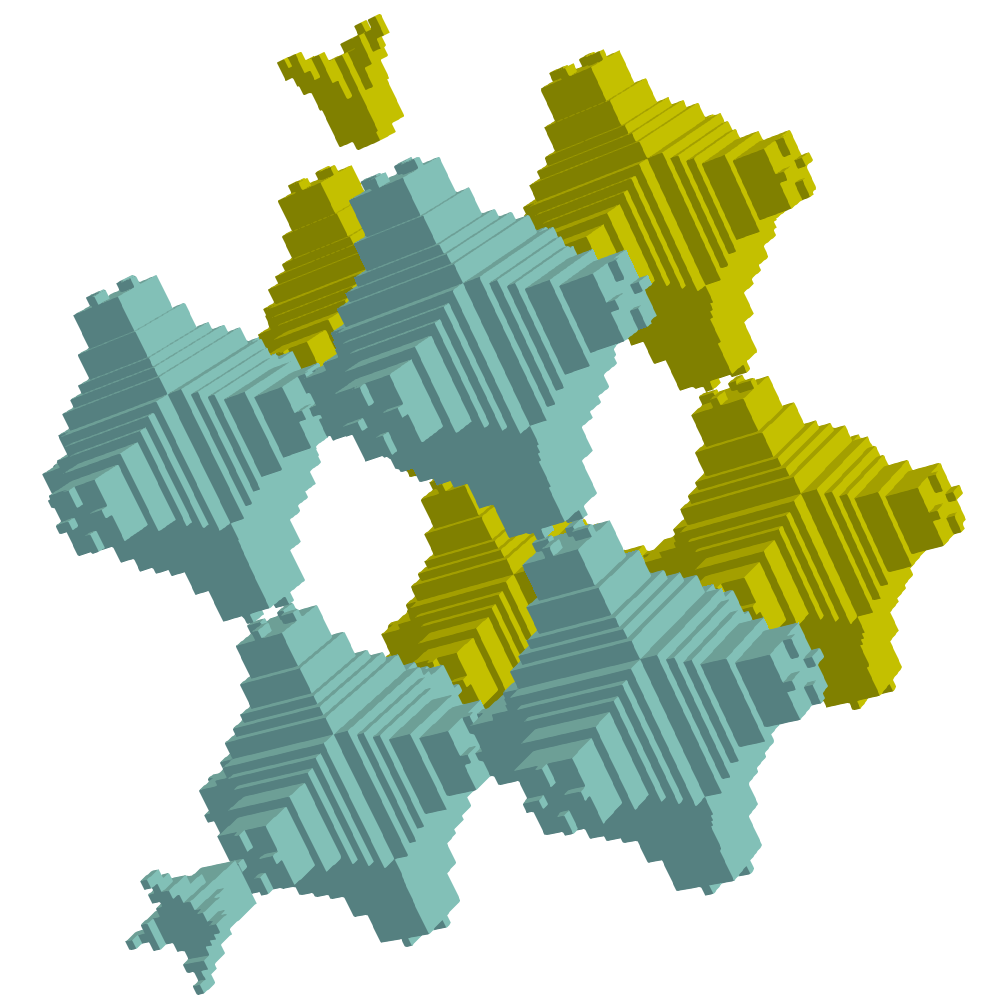}\label{fig:cubic2}}
   \hspace{0.5em}
   \subfloat[]{\includegraphics[width=0.20\linewidth]{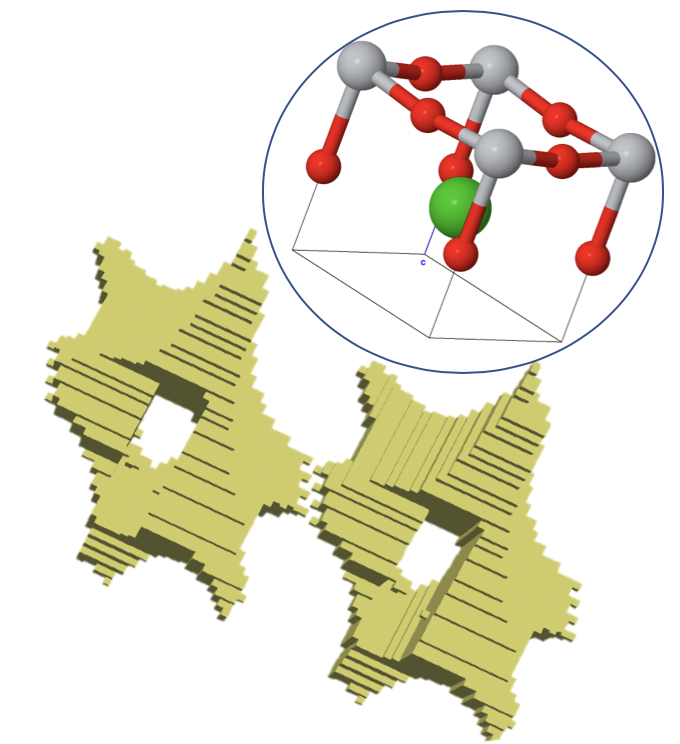}\label{fig:tet1}}
    \caption{(a)~Cubic lattice structure of $BaTiO_3$ (courtesy Springer Materials~\cite{springer-mat})
    with diffused structure in backdrop.
    (b)~Minimal persistent 2-cycles computed on the original function. 
    (c)~Minimal persistent 2-cycles computed on the negated function. 
    (d)~Minimal persistent 2-cycles computed on the negated function of a tetragonal lattice structure of $BaTiO_3$. The inlaid picture~\cite{springer-mat} illustrates the bonds of the structure. }
     \label{fig:results3}
\end{figure}

\paragraph{Material science.} 
We consider the atomic configuration of $BaTiO_3$, 
which is a ferroelectric material used for making capacitors, transducers, and microphones.
Figure~\ref{fig:cubicStruct} shows the atomic configuration of the molecule, 
where the red, grey, and green balls denote the 
Oxygen, Titanium, and Barium atoms separately
and the radii of the balls equal the radii of the corresponding atoms.
Volume data are built by uniformly sampling a $3\times3\times3$ lattice structure
similar to the one shown in Figure~\ref{fig:cubicStruct},
with the step width equal to one {\it angstrom}
(note that Figure~\ref{fig:cubicStruct} only shows a $2\times2\times2$ lattice structure).
Scalar value on a point of the volume is determined as follows:
For each atom, let the distance from the point to the atom's center be $d$,
then the scalar value of the point contributed by the atom is $\max\Set{w(r-d)/r,0}$,
where $r$ is the radius of the atom and $w$ is the atomic weight.
The scalar value on the point is then equal to the sum of the above values contributed by all atoms.
For the purpose of this experiment, we computed minimal persistent 2-cycles 
on both the original scalar function and its negated one.
Figure~\ref{fig:cubic0} shows a portion of the minimal persistent 2-cycles computed on the original function,
where the purple, red, and green cycles correspond to atoms of Barium, Titanium, and Oxygen respectively.
In our experiment, every atom corresponds to such a minimal persistent 2-cycle of a long interval.
Figure~\ref{fig:cubic2} shows a portion of the minimal persistent 2-cycles computed on the negated function, 
where the cycles complement the Barium atoms.
Figure~\ref{fig:tet1} shows the output on the negated function from a tetragonal lattice structure~\cite{springer-mat},
where the atomic bonds are not straight (see Figure~\ref{fig:tet1} inlay). 
The stretch on the lattice structure leads to minimal persistent 2-cycles with non-trivial genus.

\section{Conclusions} 
In this paper, 
we inspect the computational complexity for several problems 
concerning minimal persistent cycles.
We expand the hardness results found in~\cite{dey19pers1cyc}
and discover the cases that are NP-hard
and others that are solvable in polynomial time.
For general complexes,
we conclude that the computation is NP-hard over all dimensions for finite intervals
and NP-hard over dimension greater than one for infinite intervals.
Besides,
we find the problems to be tractable in dimension $\Dimless$ if the
given complex is a weak $(\Dimtop)$-pseudomanifold and, for infinite intervals, 
if the weak $(\Dimtop)$-pseudomanifold is embedded in $\real^\Dimtop$.

This research leads to
some open questions concerning persistent cycles:

\textsf{i.} In our experiments, 
some persistent cycles correspond to important features of the data (see Section~\ref{sec:exp}).
However, we also ran into some intervals whose persistent cycles do not
have obvious meanings.
If there are ways to design filtrations for data 
such that persistent cycles are related to the important features,
then the prospect for the application of persistent cycles
or persistence in general would be more extensive.

\textsf{ii.} As found in~\cite{dey19pers1cyc},
persistent cycles are not stable in general even when
only the weights of the cycles are considered.
It will be helpful to figure out assumptions that are
still relevant in practice, but
under which the persistent cycles remain stable.


\textsf{iii.} We have presented $O(n^2)$-time algorithms for computing a minimal
persistent cycle for a given interval. A natural question is whether this
time complexity can be improved. Furthermore, can we devise a better algorithm
to compute minimal persistent cycles for all intervals 
(i.e., the {\it minimal persistent basis}~\cite{dey19pers1cyc}), 
improving upon the obvious
$O(n^3)$-time algorithm that runs our algorithms on each interval?

\section*{Acknowledgments:} This research was conducted with the support of the NSF grants CCF-1740761 and CCF-1839252. We thank the anonymous reviewers for insightful comments.

\appendix

\section{Proof of Theorem~\ref{thm:walk-correct}}

We first define some symbols used in this section.
The interior of a set $U$ is denoted by $\interior(U)$. 
The boundary of a topological ball $\bll$ is denoted by $\bd(\bll)$.
The set of $\diml$-cofaces of a simplex $\sG$ in a $\DG$-complex~\cite{hatcher2002algebraic} $K$ is denoted by $\cof^K_{\diml}(\sG)$.
\label{apx:walk-correct-pf}

The proof of Theorem~\ref{thm:walk-correct} is based on the extended Jordan--Brouwer separation theorem (Theorem~\ref{thm:jordan-sep}) by Alexander~\cite{alexander1922proof}. 
The statement of the theorem depends on the following definition:

\begin{Definition}[Pseudomanifold] \label{dfn:pseudomnfld}
A simplicial complex $K$ is a {\it $\diml$-pseudomanifold} if $K$ is a pure $\diml$-complex and each $(\diml-1)$-simplex is a face of exactly two $\diml$-simplices in $K$.
\end{Definition}

\begin{Remark}
Note that definitions for $\diml$-pseudomanifolds, such as in~\cite{spanier1989algebraic}, typically assume the complex to be $\diml$-connected. 
\end{Remark}

\begin{theorem}
\label{thm:jordan-sep}
Let $\diml>1$ and $\Mcal$ be a finite $(\diml-1)$-connected $(\diml-1)$-pseudomanifold embedded in $\real^\diml$, 
then $\real^\diml\smallsetminus|\Mcal|$ has exactly 2 connected components.
\end{theorem}

Now we can finish our proof:

\begin{proof}[Proof of Theorem~\ref{thm:walk-correct}]
The general idea of the proof is as follows:
Using a trick which we call the ``de-contracting'', 
we first create a $\DG$-complex $\wdtild{K}'$
where each oriented simplex of $\vec{\zG}_j$ uniquely corresponds to an unoriented simplex.
Then, using a trick which we call the ``de-pinching'',
we show that $\vec{\zG}_j$ is the boundary of a region $\Acal$.
Finally, from the above fact, we use proof by contradiction to reach the conclusion.
Figure~\ref{fig:cyc-decontract} gives an example of the ``de-contracting'' and ``de-pinching''.

First, let $\SG'$ be the set of $\Dimless$-simplices of $\wdtild{K}$ whose both oriented simplices are in $\vec{\zG}_j$.
For a $\Dimless$-simplex $\sG^{\Dimless}$ of $\SG'$, 
we can let $\bll'$ be a topological $(\Dimtop)$-ball residing in $\real^\Dimtop$ 
such that $\bd(\bll')$ equals two $\Dimless$-simplices with boundaries glued together.
We then homeomorphically map points of $\real^\Dimtop\smallsetminus\sG^{\Dimless}$ to $\real^\Dimtop\smallsetminus\bll'$.
By taking care of the mapping near the boundary of $\bll'$,
we can get a new ambient $\real^\Dimtop$ and a new $\DG$-complex where
all simplices of $\wdtild{K}$ are untouched except that 
$\sG^{\Dimless}$ now corresponds to the two $\Dimless$-simplices bounding $\bll'$.
We can also think of the above process as 
``de-contracting'' the topological $\Dimless$-ball $\sG^{\Dimless}$
into the topological $(\Dimtop)$-ball $\bll'$ 
so that $\sG^{\Dimless}$ turns into two separate $\Dimless$-simplices with identical $(\Dimless-1)$-faces 
(see Figure~\ref{fig:decontract} for an example).
After doing the ``de-contraction'' for all $\Dimless$-simplices in $\SG'$,
we get a $\DG$-complex $\wdtild{K}'$.
It is true that an oriented boundary $\Dimless$-simplex in $\wdtild{K}$ 
can be naturally identified as an oriented boundary $\Dimless$-simplex in $\wdtild{K}'$.
It is also true that the groups of oriented boundary $\Dimless$-simplices in $\wdtild{K}$
are still groups of oriented boundary $\Dimless$-simplices in $\wdtild{K}'$
under the natural identification.
So we can let $\vec{\zG}_j$ denote the same group of oriented $\Dimless$-simplices in $\wdtild{K}'$.
The construction guarantees that if $\vec{\zG}_j$ is the boundary of a void of $\real^\Dimtop\smallsetminus|\wdtild{K}'|$,
then $\vec{\zG}_j$ is also the boundary of a void of $\real^\Dimtop\smallsetminus|\wdtild{K}|$.
So we only need to show that $\vec{\zG}_j$ is the boundary of a void of $\real^\Dimtop\smallsetminus|\wdtild{K}'|$ 
(see Figure~\ref{fig:cyc-decontract} for an example).
From now on, we always treat $\vec{\zG}_j$ as a set of oriented $\Dimless$-simplices as well as
a $\Dimless$-cycle (with $\Zbb$ coefficients) in $\wdtild{K}'$.

\begin{figure}[tb!]
\centering
  \subfloat[]{\includegraphics[width=0.28\linewidth]{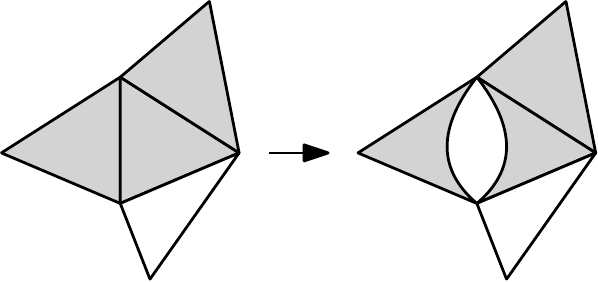}\label{fig:decontract}}
  \hspace{3em}
  \subfloat[]{\includegraphics[width=0.6\linewidth]{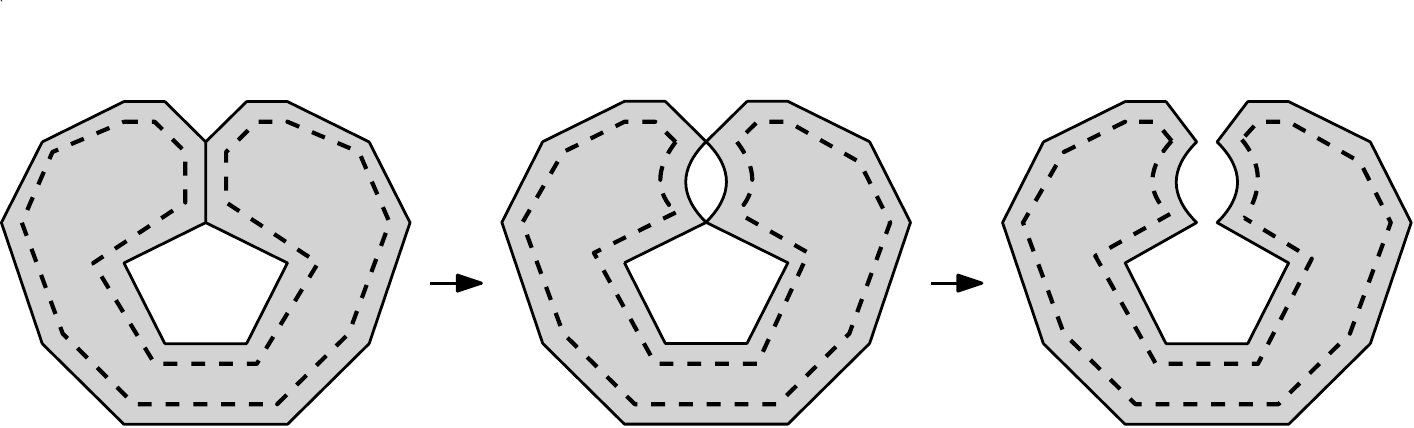}\label{fig:cyc-decontract}}
   \caption{(a) An example of the ``de-contraction'' of $\sG^{\Dimless}$ for $\Dimless=1$, 
   where a $1$-simplex in the left simplicial complex 
   turns into two curved $1$-simplices with identical boundary in the right $\DG$-complex.
   The topological 2-ball $\bll'$ is the one bounded by the two curved $1$-simplices.
   (b)~Left to middle: An example demonstrating the void boundary correspondence 
   from $\wdtild{K}$ to $\wdtild{K}'$ for $\Dimless=1$.
   After a $1$-simplex is de-contracted,
   the shaded void for $\wdtild{K}$ corresponds to the shaded void for $\wdtild{K}'$
   and their boundaries (dashed line) can be identified.
   Middle to right: The ``de-pinching'' properly separates apart incident edges (1-simplices) for 
   the two vertices (0-simplices) having more than two 1-cofaces.
   The complex $\Mcal_h$ (on the right) then becomes a pseudomanifold.
   Te deform $\Mcal_h$ back to $\Mcal$ (in this example $\Mcal=\wdtild{K}'$), 
   only points in $\Bcal$ (unshaded region) are contracted.
   }
\end{figure}

Since different oriented simplices of $\vec{\zG}_j$
correspond to different unoriented simplices in $\wdtild{K}'$,
we define a bijection $\psi:\vec{\zG}_j\to\zG$.
The bijection $\psi$ maps each oriented simplex of $\vec{\zG}_j$ 
to its corresponding unoriented simplex 
and $\zG$ is the image of this mapping.
We then let $\Mcal$ be the closure of the simplicial set $\zG$.
Note that $\zG$ is a $\Dimless$-cycle (with $\Zbb_2$ coefficients) of $\wdtild{K}'$ and $\Mcal$ is a subcomplex of $\wdtild{K}'$.
Therefore, each $(\Dimless-1)$-simplex is a face of an even number of $\Dimless$-simplices in $\Mcal$.
We first pick a $(\Dimless-1)$-simplex $\sG^{\Dimless-1}$ of $\Mcal$ such that $\big|\cof_{\Dimless}^\Mcal(\sG^{\Dimless-1})\big|>2$,
then pick two $\Dimless$-simplices $\sG^{\Dimless}_0$ and $\sG^{\Dimless}_1$ from $\cof_{\Dimless}^\Mcal(\sG^{\Dimless-1})$ 
such that $\psi\inv(\sG^{\Dimless}_0)$ and $\psi\inv(\sG^{\Dimless}_1)$ are paired 
in the void boundary reconstruction for $\wdtild{K}'$.
It is then true that 
$\sG^{\Dimless}_0\union\sG^{\Dimless}_1$ forms a topological $\Dimless$-ball $\bll^{\Dimless}_1$ 
containing $\sG^{\Dimless-1}$.
Forming the topological $\Dimless$-balls for all such pairs of $\Dimless$-simplices in $\cof_{\Dimless}^\Mcal(\sG^{\Dimless-1})$,
we get a set of $\Dimless$-balls $\Set{\bll^{\Dimless}_1,\ldots,\bll^{\Dimless}_\kG}$
for $\kG=\big|\cof_{\Dimless}^\Mcal(\sG^{\Dimless-1})\big|\big/2$.
For each $i$,
we slightly move 
$\bll^{\Dimless}_i\smallsetminus\interior(\sG^{\Dimless-1})$ 
while keeping $\bd(\bll^{\Dimless}_i)$ untouched.
We then take the closure of each $\bll^{\Dimless}_i\smallsetminus\interior(\sG^{\Dimless-1})$ 
to get a new $\DG$-complex $\Mcal_1$
in which the $\bll^{\Dimless}_i$'s have their interiors disjoint.
Note that in $\Mcal_1$, $\sG^{\Dimless-1}$ now corresponds to $\kG$ different $(\Dimlsls)$-simplices sharing the boundary.
We can repeat the above ``de-pinching'' process for each $(\Dimless-1)$-simplex having more than two $\Dimless$-cofaces in $\Mcal$
and then get a sequence of $\DG$-complexes $(\Mcal_0,\Mcal_1,\ldots,\Mcal_h)$.
In the sequence, $\Mcal_0=\Mcal$ and 
$\Mcal_i$ is derived from $\Mcal_{i-1}$ 
by doing the ``de-pinching'' on a $(\Dimless-1)$-simplex.
It is then true that $\Mcal_h$ is a pure $\Dimless$-dimensional $\Dimless$-connected $\DG$-complex
where each $(\Dimless-1)$-simplex is a face of exactly two $\Dimless$-simplices.
Since we can subdivide $\Mcal_h$ to make it a simplicial complex, 
by Theorem~\ref{thm:jordan-sep}, $|\Mcal_h|$ must separate $\real^\Dimtop$ into two connected components.
Note that for each~$i$, we can treat $\real^\Dimtop\smallsetminus|\Mcal_i|$ as a subset of $\real^\Dimtop\smallsetminus|\Mcal_{i+1}|$
because to deform $\Mcal_{i+1}$ back to $\Mcal_{i}$,
we only need to contract some points in $\real^\Dimtop\smallsetminus|\Mcal_{i+1}|$
to points in $|\Mcal_{i+1}|$.
Then 
the connected components of $\real^\Dimtop\smallsetminus|\Mcal|$ are still connected in $\real^\Dimtop\smallsetminus|\Mcal_{h}|$.
Since all oriented $\Dimless$-simplices of $\vec{\zG}_j$ bound the same void of $\real^\Dimtop\smallsetminus|\wdtild{K}'|$,
we can let this void be $\Vcal$.
The void $\Vcal$ is still connected in $\real^\Dimtop\smallsetminus|\Mcal|$ because $\real^\Dimtop\smallsetminus|\wdtild{K}'|\subseteq \real^\Dimtop\smallsetminus|\Mcal|$.
Therefore, $\Vcal$ is still connected in $\real^\Dimtop\smallsetminus|\Mcal_h|$.
We can let $\Acal$ be the connected component of $\real^\Dimtop\smallsetminus|\Mcal_h|$ containing $\Vcal$
and let $\Bcal$ be the other connected component.
The $\Dimless$-simplices in $\Mcal$ and $\Mcal_h$ can be identified
because going from each $\Mcal_i$ to $\Mcal_{i+1}$ 
the interior of each $\Dimless$-simplex is never touched.
Therefore, $\zG$ is still a $\Dimless$-cycle (with $\Zbb_2$ coefficients) in $\Mcal_h$.
We then have that the two $\Dimless$-cycles (with $\Zbb$ coefficients) in $\Mcal_h$,
which are derived from the two consistent orientations of simplices of $\zG$,
bound $\Acal$ and $\Bcal$.
Then, as one of the two $\Dimless$-cycles (with $\Zbb$ coefficients) derived from $\zG$, 
$\vec{\zG}_j$ must be the boundary of $\Acal$ or $\Bcal$ in $\Mcal_h$.
We have that $\vec{\zG}_j$ bounds $\Acal$ because $\Bcal$ does not contain points from $\Vcal$.
A fact about our construction is that
to deform each $\Mcal_{i}$ back into $\Mcal_{i-1}$, we only need to contract points in $\Bcal$.
This implies that $\Acal$ is still a void of $\real^\Dimtop\smallsetminus|\Mcal|$ with boundary $\vec{\zG}_j$ (see Figure~\ref{fig:cyc-decontract} for an example).

To prove that $\vec{\zG}_j$ is the boundary of a void of $\real^\Dimtop\smallsetminus|\wdtild{K}'|$,
we only need to show that there are no oriented $\Dimless$-simplices which are in the boundary of $\Vcal$ but do not belong to $\vec{\zG}_j$.
For contradiction, suppose that there is such an oriented $\Dimless$-simplex $\vec{\sG}^{\Dimless}$.
Then $\vec{\sG}^{\Dimless}$ must not be oppositely oriented to any oriented simplex of $\vec{\zG}_j$ because 
otherwise $\vec{\sG}^{\Dimless}$ would bound another connected component of $\real^\Dimtop\smallsetminus|\Mcal|$ 
and thus bound another connected component of $\real^\Dimtop\smallsetminus|\wdtild{K}'|$.
Let ${\sG}^{\Dimless}$ be the unoriented $\Dimless$-simplex of $\vec{\sG}^{\Dimless}$,
then ${\sG}^{\Dimless}\not\in \Mcal$ because otherwise $\vec{\sG}^{\Dimless}$ would be oppositely oriented to an oriented simplex of $\vec{\zG}_j$.
Since ${\sG}^{\Dimless}\not\in \Mcal$, the interior of ${\sG}^{\Dimless}$ must reside in $\real^\Dimtop\smallsetminus|\Mcal|$.
From now on, we always treat $\Acal$ as a void of $\real^\Dimtop\smallsetminus|\Mcal|$.
Then among all voids of $\real^\Dimtop\smallsetminus|\Mcal|$, 
the interior of ${\sG}^{\Dimless}$ resides in $\Acal$.
This is because $\Acal$ is the void of $\real^\Dimtop\smallsetminus|\Mcal|$ containing $\Vcal$.
If ${\sG}^{\Dimless}$ resides in a void other than $\Acal$, points to either side
of ${\sG}^{\Dimless}$ cannot be from $\Vcal$.
Since $\wdtild{K}'$ is $\Dimless$-connected, 
there must be a sequence of $\Dimless$-simplices $({\sG}^{\Dimless}_0,\ldots,{\sG}^{\Dimless}_l)$ of $\wdtild{K}'$
such that ${\sG}^{\Dimless}_0={\sG}^{\Dimless}$, ${\sG}^{\Dimless}_l\in\Mcal$, 
and ${\sG}^{\Dimless}_i$, ${\sG}^{\Dimless}_{i+1}$ share a $(\Dimless-1)$-face for each $i$ such that $0\leq i<l$.
Because the interior of ${\sG}^{\Dimless}_l$ is not in $\Acal$,
we can let ${\sG}^{\Dimless}_{l'}$ be the first $\Dimless$-simplex 
in the sequence whose interior is not in $\Acal$,
then $l'\neq 0$ and the interior of ${\sG}^{\Dimless}_{l'-1}$ is in $\Acal$. 
Let ${\sG}^{\Dimless-1}_{l'-1}$ be the $(\Dimless-1)$-face shared by ${\sG}^{\Dimless}_{l'-1}$ and ${\sG}^{\Dimless}_{l'}$, 
we claim that ${\sG}^{\Dimless-1}_{l'-1}\in\Mcal$.
If ${\sG}^{\Dimless}_{l'}\in\Mcal$, then it is obvious that ${\sG}^{\Dimless-1}_{l'-1}\in\Mcal$.
If ${\sG}^{\Dimless}_{l'}\not\in\Mcal$, then it is also true that ${\sG}^{\Dimless-1}_{l'-1}\in\Mcal$ 
because otherwise the interiors of ${\sG}^{\Dimless}_{l'-1}$ and ${\sG}^{\Dimless}_{l'}$ would be connected in $\real^\Dimtop\smallsetminus|\Mcal|$.
Around the neighborhood of ${\sG}^{\Dimless-1}_{l'-1}$ during the void boundary reconstruction for $\wdtild{K}'$,
any two paired oriented simplices from $\vec{\zG}_j$ enclose a region residing in $\Acal$.
Because of the nature of the pairing, 
${\sG}^{\Dimless}_{l'-1}$ cannot be contained in any of the regions enclosed 
by the paired oriented simplices from $\vec{\zG}_j$.
Since $\vec{\zG}_j$ is the boundary of the void $\Acal$ of $\real^\Dimtop\smallsetminus|\Mcal|$,
all other regions in the neighborhood of ${\sG}^{\Dimless-1}_{l'-1}$ must not be in $\Acal$.
This implies that ${\sG}^{\Dimless}_{l'-1}$ is not in $\Acal$, which is a contradiction.
\end{proof}

\bibliographystyle{plain} 
\bibliography{refs}

\end{document}